%% file: main.tex
\documentclass[pdflatex,sn-mathphys-num]{sn-jnl}

\usepackage{ifpdf}
\usepackage{amsmath}
\usepackage{amsthm}
\usepackage{algorithm}
\usepackage{color}
\usepackage[noend]{algpseudocode}
\usepackage{algorithmicx}
\usepackage{array}
\usepackage{epsfig}
\usepackage{graphicx}
\usepackage{epstopdf}
\usepackage{amssymb}
\usepackage{array}
\usepackage{xcolor}
\usepackage{booktabs}
\usepackage{multirow}
\usepackage{algorithm}
\usepackage{tikz}
\usetikzlibrary{calc}
\usetikzlibrary{shapes,arrows,positioning,decorations.markings}
\usetikzlibrary{shapes.geometric, arrows}
\usepackage{relsize}
\usepackage{setspace}
\usepackage{hhline}
\usepackage{pgfplots, pgf}
\usepackage{subfigure}
\usepackage{caption}
\usepackage{placeins}
\tikzset{fontscale/.style = {font=\relsize{#1}}}

\newtheorem{prop}{Proposition}

\newcolumntype{P}[1]{>{\centering\arraybackslash}p{#1}}

\usepackage{graphicx}%
\usepackage{multirow}%
\usepackage{amsmath,amssymb,amsfonts}%
\usepackage{amsthm}%
\usepackage{mathrsfs}%
\usepackage[title]{appendix}%
\usepackage{xcolor}%
\usepackage{textcomp}%
\usepackage{manyfoot}%
\usepackage{booktabs}%
\usepackage{algorithm}%
\usepackage{algorithmicx}%
\usepackage{algpseudocode}%
\usepackage{listings}%
\usepackage{placeins}  
\usepackage{float}      
\usepackage{needspace}
\usepackage{flafter}

\theoremstyle{thmstyleone}%
\theoremstyle{thmstyletwo}%

\theoremstyle{thmstylethree}%

\begin{document}

\title[High-Accuracy and Efficient DV-Hop Localization for IoT Using Hop Loss]{High-Accuracy and Efficient DV-Hop Localization for IoT Using Hop Loss}


\author*[1]{\fnm{Zhengdi} \sur{Shen}}\email{zs267@cornell.edu}

\author[2]{\fnm{Qiran} \sur{Wang}}\email{qw362@nyu.edu}
\equalcont{These authors contributed equally to this work.}

\affil*[1]{\orgdiv{Center for Applied Mathematics}, \orgname{Cornell University}, \orgaddress{\city{Ithaca}, \postcode{14850}, \state{NY}, \country{USA}}}

\affil[2]{\orgdiv{Courant Institue of Mathematical Sciences}, \orgname{New York University}, \orgaddress{\city{New York}, \postcode{10012}, \state{NY}, \country{USA}}}

\input{sections/abstract.tex}

\maketitle
\footnotetext{This manuscript has been submitted to \emph{Cluster Computing} and is under peer review.}
\newpage

\input{sections/introduction.tex}

\input{sections/related_work}

\input{sections/connectivity_loss}

\input{sections/experiments}

\input{sections/conclusion}

\bibliography{sn-bibliography}

\end{document}

%% file: sections/abstract.tex
\abstract{

Accurate localization is critical for Internet of Things (IoT) applications.
Using hop loss in DV-Hop-based algorithms is a promising approach. Nevertheless, challenges lie in overcoming the computational complexity caused by re-calculating the predicted hop-counts, and how to further optimize the modeling for better accuracy.
In this paper, a novel hop loss modeling, {distance-based connectivity consistency} (DCC), is proposed. 
By focusing on the first order connectivity, DCC avoids computing predicted hop-counts, and significantly reduces the time complexity. 
We also provide a proof to theoretically guarantee that this design achieves a full coverage of all hop errors. In addition, by computing a continuous loss function instead of the discrete hop-count errors, DCC further improves the localization accuracy. In the evaluations, DCC demonstrates notable improvements in accuracy over other highly regarded algorithms, and reduces 30\% to 40\% total computation time compared with the baseline algorithm using hop loss.

}

\keywords{
IoT localization, DV-Hop, hop loss, multi-objective optimization.
}

%% file: sections/introduction.tex
\section{Introduction}\label{sec:introduction}

The Internet of Things (IoT) refers to a system of interconnected physical items that utilize sensors, software, and various technologies to facilitate data exchange with other devices and systems through the internet. IoT enables smart devices to interact and gather information autonomously, enhancing efficiency in fields like healthcare, transportation, and industry \cite{aghdam2021role,ushakov2022internet,malik2021industrial,wu2024scalable}.
For IoT, localization algorithms \cite{ghorpade2021survey} play a critical role in determining the physical locations of devices. These algorithms can be divided into two categories:
\begin{enumerate}
    \item Range-based algorithms: These rely on direct measurements such as distance, angle, or time of arrival (e.g., using GPS or signal strength). Examples include Time of Arrival (TOA)\cite{xiong2021maximum}, Angle of Arrival (AOA)\cite{friedrichs2023angle}, and Received Signal Strength Indicator (RSSI)\cite{alanezi2021range} methods, which calculate the location based on precise range measurements between nodes \cite{gresham2004ultra, bavcic2024jy61}.
    \item Range-free algorithms: These do not require exact distance measurements and instead use connectivity or hop-based methods. Examples include DV-Hop \cite{niculescu2003dv09} and Centroid algorithms\cite{gupta2020study,liu2024range}, which estimate the relative location of nodes using information like the number of hops or the network topology, making them more cost-effective and scalable for large IoT networks.
\end{enumerate}
Range-free methods are often preferred for large-scale, cost-sensitive IoT applications where precise ranging is difficult. DV-Hop (Distance Vector-Hop) algorithm \cite{niculescu2003dv09} is a well-known range-free localization method used in IoT applications. It is based on the idea of estimating distances between nodes by counting the number of hops (hop-counts) between them, and then solve the node locations to fit the distance estimation. Fig. \ref{fig:dv_hop_illustration} is an illustration of it. DV-Hop algorithm requires some nodes with known locations in the network, named as anchor nodes. And other nodes, named as unknown nodes, can be localized based on the following key steps:
\begin{enumerate}
    \item Hop-count calculation: each node broadcasts its location to neighboring nodes with distances no larger than a communication radius, and this information is propagated through the network. As the message passes through intermediate nodes, each node increments the hop count, which records the number of hops from one node to another.
    \item Distance estimation: the physical distance between each unknown node and each anchor node is estimated based on the hop-counts in the network.
    \item Location calculation: the unknown node uses multilateration, a geometric method, to estimate its location based on the calculated distances to multiple anchor nodes.
\end{enumerate}

\input{images/dv_hop_illustration.tex}

Although the DV-Hop algorithm provides an effective way to estimate the node locations, its accuracy can be improved by further exploiting the hop-count information between nodes. \cite{wang2024dv} proposes a multi-objective DV-Hop-based algorithm with one objective as the hop loss, which reflects the discrepancy between the real hop-counts and the predicted hop-counts. The predicted hop-counts are the hop-counts derived from the predicted node locations. The hop loss allows for a more straightforward and comprehensive utilization of the hop-count information in the networks, and demonstrates impressive accuracy improvements in the experiments.

While the hop loss proposed in \cite{wang2024dv} demonstrates high efficacy, there remain several aspects where enhancements can be made, particularly in two points.
First, the existing algorithm requires calculating the predicted hop-counts at each step of the optimization process. These hop-count calculations require calling shortest-path algorithms in a graph \cite{madkour2017survey}, which is time expensive. Therefore, reducing the computational complexity for algorithms with hop loss is essential in large-scale networks or real-time applications where efficiency is critical.
Second, because how the hop loss is modeled has a significant influence on the performance of the algorithm, a further study to optimize the design of the hop loss function could lead to more accurate algorithms.

In this paper, a distance-based connectivity consistency (DCC) hop loss is proposed, which reduces the time complexity and achieves better localization accuracy. In the design of this loss function, two points are focused: i) the activation condition (AC) between each pair of nodes, which means whether the hop error between the nodes should be included in the total hop loss in the network, and ii) what the loss function should be for each single pair of nodes, named as individual loss (IL). Our main contributions include:
\begin{enumerate}
    \item We propose an AC based on the node connectivity consistency (CC), named as $AC^{CC}$. We theoretically prove that $AC^{CC}$ has a full coverage of all hop errors in networks, which is not achieved by the existing algorithm using hop loss. In addition, $AC^{CC}$ also has reduced time complexity.
    \item We propose a continuous IL based on the distance between nodes, named as $IL^{DST}$. Because $IL^{DST}$ is a continuous mapping from the node locations to the loss, and it does not require the computation of predicted hop-counts in the optimization process, $IL^{DST}$ provides advantages in both accuracy and time complexity compared with the existing hop loss.
    \item We run simulations with the algorithm using DCC hop loss, which combines $AC^{CC}$ and $IL^{DST}$, and compare the time complexity and accuracy between the results of our algorithm and those of the other algorithms to demonstrate the effectiveness of the proposed method.
\end{enumerate}

The rest of the paper is organized as follows:
Section \ref{sec:related_work} summarizes the work related to DV-Hop based algorithms.
Section \ref{sec:connectivity_loss} introduces the proposed hop loss function as one of the objectives in the algorithm.
Section \ref{sec:exp} presents the evaluations of the simulation results.
Section \ref{sec:conclusion} is the conclusion and future work of this paper.

%% file: images/dv_hop_illustration.tex
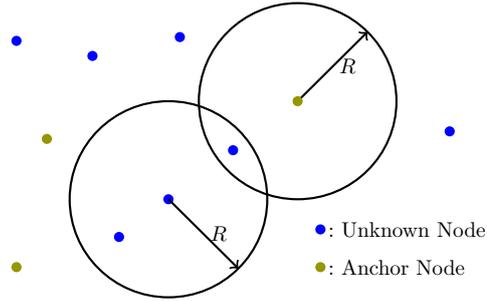
\begin{figure}[h!]
    \centering

    \definecolor{anchor_color}{RGB}{150,150,0}
    \definecolor{unknown_color}{RGB}{0,0,255}
        
    \begin{tikzpicture}[thick,scale=1, every node/.style={scale=0.8}]
        
        \def\numberOfDotsX{3}
        \def\numberOfDotsY{3}

        \pgfmathsetmacro\radius{1.3}
        \pgfmathsetmacro{\anchorx}{1.2};
        \pgfmathsetmacro{\anchory}{0.7};
        \pgfmathsetmacro{\unkx}{-0.5};
        \pgfmathsetmacro{\unky}{-0.6};
        \coordinate (anchor_dot) at (\anchorx, \anchory);
        \coordinate (anchor_edge) at (\anchorx + \radius / 1.414, \anchory + \radius / 1.414);
        \coordinate (unk_dot) at (\unkx, \unky);
        \coordinate (unk_edge) at (\unkx + \radius / 1.414, \unky - \radius / 1.414);

        \def\dotSize{1.5pt} 
        \def\legendDotSize{1.5pt}
        
        \draw[black] (anchor_dot) circle (\radius);
        \node[right, font=\large\bfseries] at (anchor_dot) {};
        \draw[->](anchor_dot) -> (anchor_edge);
        \node at ($(anchor_dot)!0.5!(anchor_edge) + (0.2,0.0)$) {
        $R$};
        \filldraw[anchor_color] (anchor_dot) circle (\dotSize);

        \filldraw[unknown_color] (unk_dot) circle (\dotSize);
        \draw[black] (unk_dot) circle (\radius);
        \draw[->](unk_dot) -> (unk_edge);
        \node at ($(unk_dot)!0.5!(unk_edge) + (0.2,0.0)$) {$R$};
        \node[right, font=\large\bfseries] at (unk_dot) {};

        \filldraw[unknown_color] ($(anchor_dot)!0.5!(unk_dot)$) circle (\dotSize);

        \filldraw[anchor_color] (-2.5, -1.5) circle (\dotSize);
        \filldraw[anchor_color] (-2.1, 0.2) circle (\dotSize);

        \filldraw[unknown_color] (-1.5, 1.3) circle (\dotSize);
        \filldraw[unknown_color] (3.2, 0.3) circle (\dotSize);
        \filldraw[unknown_color] (-1.15, -1.1) circle (\dotSize);
        \filldraw[unknown_color] (-2.5, 1.5) circle (\dotSize);
        \filldraw[unknown_color] (-0.35, 1.55) circle (\dotSize);

    \coordinate (unk_legend_dot) at (1.5, -1);  
    \filldraw[unknown_color] (unk_legend_dot)  circle (\legendDotSize); 
    \draw (unk_legend_dot) node[right] {: Unknown Node};
    \coordinate (anchor_legend_dot) at (1.5, -1.5);  
    \filldraw[anchor_color] (anchor_legend_dot)  circle (\legendDotSize);
    \draw (anchor_legend_dot) node[right] {: Anchor Node};
    
    \end{tikzpicture}

    \caption{An example of the DV-Hop localization. A node has a communication radius $R$, and can directly communicate with other nodes within the communication radius. The hop-count between two directly connected nodes is 1, and the hop-counts between other nodes can be derived by information propagation.}
    \label{fig:dv_hop_illustration}
\end{figure}

%% file: sections/related_work.tex
\section{Related Work}\label{sec:related_work}

The DV-Hop algorithm is a widely used localization technique for IoT. It solves the locations of unknown nodes by estimating the distances between nodes with the hop-counts in the networks. The major steps of the DV-Hop algorithm include broadcasting messages to obtain the hop-counts between nodes, estimating the distance between each anchor node and unknown node, and solving the objective function(s) for localization. The original DV-Hop algorithm is proposed by \cite{niculescu2003dv}, in which an average distance is estimated for each anchor node $i$ as
\begin{equation}
    AvgDis_i = \frac{\sum_{j\in Anchors}\sqrt{(x_i - x_j)^2 + (y_i - y_j)^2}}{\sum_{j\in Anchors}hop_{i, j}},
\end{equation}
where $(x_i, y_i)$ and $(x_j, y_j)$ are anchor node locations, and $hop_{i, k}$ is the hop-count between node $i$ and $j$. Then the distance between an anchor node $i$ and an unknown node $k$ is estimated by
\begin{equation}
    Dis_{i, k} = AvgDis_i \cdot hop_{i, k},
\end{equation}
where $hop_{i, k}$ is the hop-count between $i$ and $k$.



Improvements based on \cite{niculescu2003dv} are developed by scholars to enhance the localization accuracy. \cite{wang2022distance, wang2024dv, wang2024probability, cui2018high, xiao2009reliable, wang2022novel} propose more accurate distance estimations between nodes. \cite{sun2023improvement, jia2022high, wang2022distance, OUYANG@023_2, cui2017novel023, cao2023dv, zhang2024dv, mani2024three, kanwar2021dv} propose bio-inspired optimization algorithms, including genetic algorithms (GAs), to enable the optimization for general forms of objective functions. \cite{xiao2009reliable, yan2018dv, wang2024ubeti, tang2022improved, hu2023nrap, tomic2016improvements} propose accurate algorithms for networks in irregular areas. \cite{rayavarapu2024moans,wang2022distance, wang2024dv, wang20243d} propose models with multi-objective functions to achieve solutions with better accuracy. \cite{gui2020connectivity, wang2024dv} propose objectives incorporating the consistency between the hop-counts derived from the estimated node locations and the real hop-counts. \cite{messous2020online, liouane2023new} propose online DV-Hop algorithms. \cite{wang2023dv} uses the topological similarity derived from hop-counts to improve the localization accuracy. \cite{li2020research} proposes algorithm with dual communication radius to improve the accuracy. \cite{liu2022improved} eliminates the residual errors by neural dynamics.

Among these studies, \cite{wang2024dv} proposes a multi-objective algorithm, with the first objective as the distance estimation loss, and the second objective as the hop loss. The hop loss in \cite{wang2024dv} is used as the baseline in our paper. Denote the real hop-count and the predicted hop-count between node $i$ and $j$ as $Hop_{i, j}^{real}$ and $Hop_{i, j}^{pred}$ respectively, and denote the total hop loss in the network as $HL^{base}$. $HL^{base}$ is defined as
\begin{align}\label{eq:base_hl}
    HL^{base} = \sum_{i=1}^N\sum_{j=1}^N(Hop_{i,j}^{real}-Hop_{i,j}^{pred})^2  
    \cdot\mathbf{1}_{\{Hop_{i, j}^{real}<3\}},
\end{align}
where $N$ is the number of all nodes, and $\mathbf{1}_{\{\cdots\}}$ is the characteristic function.
In equation (\ref{eq:base_hl}), the factor $(Hop_{i,j}^{real}-Hop_{i,j}^{pred})^2$ reflects the discrepancy between the real hop-count and the predicted hop-count for each pair of nodes, and the factor $\mathbf{1}_{\{Hop_{i, j}^{real}<3\}}$ filters out the node pairs where the real hop-counts are too large to get better accuracy. The threshold 3 is obtained empirically.
Equation (\ref{eq:base_hl}) effectively leverages the hop-counts to enhance the performance. Particularly, the information between any two unknown nodes and the information between any two anchor nodes are also used in this objective, while the objectives about distance loss only focus on the relationship between unknown nodes and anchor nodes. The exploitation of the information between arbitrary nodes significantly improves the localization accuracy. Meanwhile, there remain opportunities for further enhancement:
\begin{enumerate}
    \item The predicted hop-counts calculation in every optimization iteration is computationally expensive. This makes the time complexity of the algorithm using hop loss larger than other algorithms purely based on distance losses. Reducing the time complexity of hop-loss calculations is a topic worth studying.
    
    \item The current algorithm does not catch hop errors where $Hop_{i, j}^{real} \geq 3$. Although this filtering improves the algorithm performance compared with penalizing all hop errors, it could also lead to a false negative hop loss in certain cases. Enhancements can be made to more effectively penalize the hop errors and improve the accuracy.
    \item The current hop loss is a piece-wise constant discontinuous function with respect to the predicted node locations. Although this formulation is straightforward, small changes in the predicted locations may not be reflected in this hop loss, making it less sensitive to fine-grained variations in node positions during the optimization process. 
    A loss function with better continuity could more effectively guide the algorithm towards accurate solutions.
\end{enumerate}

%% file: sections/connectivity_loss.tex
\section{High Accuracy Hop Loss with Optimized Efficiency}\label{sec:connectivity_loss}

In this section, a new hop loss function is proposed to improve the localization accuracy and time complexity. Particularly, two questions are studied in the new function:
\begin{enumerate}
    \item In what conditions should the loss between two nodes contribute to the total loss function? In this paper, this condition is named as \textbf{activation condition} (AC).
    \item How should the hop discrepancy between two nodes be quantified? In this paper, the hop discrepancy from a single pair of nodes is named as \textbf{individual loss} (IL).
\end{enumerate}


\begin{table}[h!]
    \centering
    \renewcommand{\arraystretch}{1.15}
    \caption{Notations used in the hop loss calculation.}
    \begin{tabular}{|p{1.2cm}|p{11cm}|}
    \hline
       \textbf{Notation}  & \textbf{Description}  \\ \hline 
       $HL$ & The total hop loss in the network. We use $HL^{<\cdots>}$ to distinguish the hop losses of different algorithms. \\ \hline
       $i, j$ & The node indices.\\ \hline
       $AC_{i,j}$  & A binary function for the activation condition between node $i$ and $j$. $AC_{i,j}=1$ if the activation condition is true. Otherwise $AC_{i,j}=0$.We use $AC_{i,j}^{<\cdots>}$ to distinguish the activation conditions of different algorithms. \\ \hline
       $IL_{i,j}$  & The individual loss between node $i$ and $j$. $IL_{i,j}\geq 0$. We use $IL_{i,j}^{<\cdots>}$ to distinguish the individual losses of different algorithms. \\ \hline
       $Hop_{i,j}^{real}$ & The real hop-count between node $i$ and $j$. \\ \hline
       $Hop_{i,j}^{pred}$ & The hop-count between node $i$ and $j$ derived from the predicted node locations. \\ \hline
       $Dist_{i,j}^{real}$ & The real Euclidean distance between node $i$ and $j$. \\ \hline
       $Dist_{i,j}^{pred}$ & The predicted Euclidean distance between node $i$ and $j$. \\ \hline
       $N$  & The number of nodes in the network. \\ \hline
       $R$ & The communication radius. \\
    \hline
    \end{tabular}
    \label{tab:notations}
\end{table}

The notations needed in this discussion are listed in TABLE \ref{tab:notations}. The general form of the total hop loss in the network can be formulated as
\begin{equation}\label{eq:general_hl}
HL = \sum_{i=1}^{N}\sum_{j=1}^{N}IL_{i,j}\cdot AC_{i,j}.
\end{equation}
Particularly, the baseline hop loss introduced in (\ref{eq:base_hl}) can be rewritten in the form of (\ref{eq:general_hl}) as
\begin{align}
\label{eq:hl_base_2}
HL^{base} =  \sum_{i=1}^{N}\sum_{j=1}^{N}IL_{i,j}^{base}\cdot AC_{i,j}^{base},
\end{align}
where
\begin{align}
\label{eq:ac_base}
AC_{i,j}^{base}  = & \mathbf{1}_{\{Hop_{i,j}^{real} < 3\}}, \\ 
\label{eq:il_base}
IL_{i,j}^{base}  = & (Hop_{i,j}^{real} - Hop_{i,j}^{pred})^2.
\end{align}

The following of this section discusses how to modify $AC_{i,j}$ and $IL_{i,j}$ respectively in the proposed algorithm to achieve better accuracy and efficiency.

\subsection{Activation Condition Based on Connectivity Consistency}\label{ss:ac}
In the proposed AC, named as $AC^{CC}_{i,j}$, the \emph{connectivity consistency} (CC) between nodes is focused. Connectivity refers to whether the hop-count between two nodes is 1. Unlike $AC^{base}_{i,j}$ which activates node pairs with a small enough $Hop^{real}_{i,j}$, $AC^{CC}_{i,j}$ activates node pairs which have inconsistent real connectivity and predicted connectivity:
\begin{equation}\label{eq:ac_CC}
    AC_{i,j}^{CC} = \left\{
    \begin{array}{ll}
        1, & \text{if } (Hop^{real}_{i,j}=1 \text{ and } Hop^{pred}_{i,j} > 1)  \text{ or } (Hop^{real}_{i,j}>1 \text{ and }  Hop^{pred}_{i,j} = 1),\\
        0, & \text{otherwise.}
    \end{array}
    \right.
\end{equation}

TABLE \ref{tab:ac_compare} compares $AC^{base}_{i,j}$ and $AC^{CC}_{i,j}$ for different cases. Notice that in Case 2 where  $Hop^{real}_{i,j}\geq 3, Hop^{pred}_{i,j} = 1$, we have $AC^{base}_{i,j} = 0$ while $AC^{CC}_{i,j} = 1$. This means that $AC^{CC}_{i,j}$ has the advantage to catch such hop errors which cannot be caught by the baseline method. Fig. \ref{fig:ac_base_fail_case} illustrates an example corresponding to Case 2, where the prediction has a hop error between node 1 and 4. The baseline method cannot detect this inconsistency because of the design of $AC^{base}_{i,j}$, resulting in $HL^{base} = 0$, therefore cannot effectively penalize this prediction error. On the other hand, $AC^{CC}_{i,j}$ can detect this discrepancy and guide the algorithm to search for a more accurate solution.

\begin{table}[h!]
    \centering
    \caption{The comparison between $AC^{base}_{i,j}$ and $AC^{CC}_{i,j}$ for different cases.}
    \label{tab:ac_compare}
    \renewcommand{\arraystretch}{1.35}
    \begin{tabular}{|c|p{6cm}|c|c|}
         \hline
         \textbf{Case \#} & \textbf{Condition} & $\mathbf{AC^{base}_{i,j}}$ & $\mathbf{AC^{CC}_{i,j}}$ \\ \hline
         \textbf{1} & $Hop^{real}_{i,j}=1$ and $Hop^{pred}_{i,j} > 1$ & 1 & 1 \\ \hline
         \textbf{2} & $Hop^{real}_{i,j} \geq 3$ and $Hop^{pred}_{i,j} = 1$ & 0 & 1 \\ \hline
         \textbf{3} & $Hop^{real}_{i,j}=2$ and $Hop^{pred}_{i,j} > 2$ & 1 & 0 \\ \hline
         \textbf{4} & $Hop^{real}_{i,j} \geq 3 $ and $Hop^{pred}_{i,j} > 1$ & 0 & 0 \\ \hline
    \end{tabular}
\end{table}

\input{images/ac_base_fail_case/ac_base_fail_case.tex}

The provided example demonstrates the effectiveness of the proposed activation condition in a particular case. In what follows, a proof is made to show that any hop-count discrepancy can be reflected in $AC^{CC}_{i,j}$, although  $AC^{CC}_{i,j}$ only explicitly activates the first-order hop inconsistency.

\begin{prop}\label{prop:ac_CC}
  If there are two nodes $i, j$ in a network making $Hop^{real}_{i,j} \neq Hop^{pred}_{i,j}$, then there exist $i', j'$ in the network making $AC^{CC}_{i', j'}=1$.
\end{prop}

\begin{proof}
  First, consider the cases where $Hop^{real}_{i,j} > Hop^{pred}_{i,j}$. Let $h=Hop^{pred}_{i,j}$. Denote the nodes along the predicted shortest path from $i$ to $j$ as $N_0, \cdots, N_h$, where $N_0=i, N_h=j$. We have $Hop^{pred}_{N_k, N_{k+1}}=1$ for $k=0, \cdots, h-1$. On the other hand,
  \begin{equation}
      \sum_{k=0}^{h-1} Hop^{real}_{N_k, N_{k+1}} \geq Hop^{real}_{N_0,N_h}=Hop^{real}_{i,j} > H^{pred}_{i,j} = h.
  \end{equation}
  Therefore, there exists at least one $0\leq k < h$ making
  \begin{equation}
      Hop^{real}_{N_k, N_{k+1}} > 1 = Hop^{pred}_{N_k, N_{k+1}}.
  \end{equation}
  According to (\ref{eq:ac_CC}), $AC^{CC}_{N_k, N_{k+1}} = 1$. The proposition holds in this condition.

  Similarly, consider the cases where $Hop^{real}_{i,j} < Hop^{pred}_{i,j}$. Let $h=Hop^{real}_{i,j}$. Denote the nodes along the real shortest path from $i$ to $j$ as $N_0, \cdots, N_h$, where $N_0=i, N_h=j$.  We have $Hop^{real}_{N_k, N_{k+1}}=1$ for $k=0, \cdots, h-1$. On the other hand,
  \begin{equation}
      \sum_{k=0}^{h-1} Hop^{pred}_{N_k, N_{k+1}} \geq Hop^{pred}_{N_0,N_h}=Hop^{pred}_{i,j} > h.
  \end{equation}
  Therefore, there exists at least one $0\leq k < h$ making
  \begin{equation}
      Hop^{pred}_{N_k, N_{k+1}} > 1 = Hop^{real}_{N_k, N_{k+1}}.
  \end{equation}
  According to (\ref{eq:ac_CC}), $AC^{CC}_{N_k, N_{k+1}} = 1$. The proposition holds in this condition.

  Combining the two conditions, the proposition is proved.
\end{proof}

Proposition \ref{prop:ac_CC} guarantees that as long as there is a hop error, there will be a non-zero activation condition for some node pairs. Therefore, if the individual loss is designed properly, the total hop loss will be non-zero and effectively reflect the hop-count discrepancy. This also indicates that although $AC^{CC}_{i,j}$ deactivates some conditions which are active in $AC^{base}_{i,j}$, like Case 3 in TABLE \ref{tab:ac_compare}, those hop discrepancies will still be implicitly reflected in the proposed hop loss. 

The simplicity of computing the connectivity is another advantage of $AC^{CC}_{i,j}$. Obtaining the connectivity only requires comparing the node distance with $R$, and does not need calculating the hop-counts in the network, which is time expensive. The simplicity of computing $AC^{CC}_{i,j}$ lays the foundation for a time efficient algorithm.

\subsection{Individual Loss Based on Distance between Nodes}
The motivation for defining a new IL lies in addressing two aspects of the baseline method.
First, the primary objective of localization algorithms is to determine node locations, which are continuous points in the coordinate space. In contrast, hop-count is a discontinuous, piece-wise constant function that maps the continuous coordinates to discrete hop-count values. This mapping could lead to scenarios where minor changes in coordinates cannot be reflected in the hop-count related functions. By designing the hop loss as a continuous function of the node coordinates, the sensitivity of the model to these minor variations can be enhanced.
Second, the computation of hop-counts involves a high time complexity. In contrast, the calculation of the Euclidean distances between nodes is much faster. If the hop loss is defined to focus on a function of the distances between nodes, we can circumvent the need for hop-count calculations.
These two considerations motivate the proposal of an IL based on the distance between nodes, aiming for improved accuracy and efficiency in the localization process. This IL is named as $IL^{DST}_{i,j}$, where {DST} stands for \emph{distance}.

In the proposed algorithm, $IL^{DST}_{i,j}$ is directly tied to $AC^{CC}_{i,j}$ defined in (\ref{eq:ac_CC}). Therefore, it only needs to be defined on conditions where $AC^{CC}_{i,j} \neq 0$. $IL^{DST}_{i,j}$ is defined in the following two cases to guide the predicted distances between nodes toward their real values:
\begin{enumerate}
    \item If $Hop^{real}_{i,j}=1$, $Hop^{pred}_{i,j}>1$: in this case, we have $Dist^{pred}_{i,j}>R \geq Dist^{real}_{i,j}$. Define
    \begin{equation}\label{eq:il_dst1}
        IL^{DST}_{i,j} = Dist^{pred}_{i,j} - R. 
    \end{equation}
    \item $Hop^{real}_{i,j}>1$, $Hop^{pred}_{i,j}=1$: in this case, we have $Dist^{pred}_{i,j}\leq R < Dist^{real}_{i,j}$. Define
    \begin{equation}\label{eq:il_dst2}
        IL^{DST}_{i,j} = R - Dist^{pred}_{i,j}. 
    \end{equation}
\end{enumerate}
Combining (\ref{eq:il_dst1}) and (\ref{eq:il_dst2}), we have
\begin{equation}\label{eq:il_dst}
    IL^{DST}_{i,j} = |Dist^{pred}_{i,j} - R| \text{ when } AC^{CC}_{i,j} \neq 0.
\end{equation}

\input{images/ILDST/ILDST}

Fig. \ref{fig:ildst} plots the relation between $Dist^{pred}_{i,j}$ and $IL^{DST}_{i,j}$. For $Dist^{pred}_{i,j}$, the direction toward $Dist^{real}_{i,j}$ is always consistent with the direction toward smaller $IL^{DST}_{i,j}$. Therefore, minimizing the proposed hop loss can guide the predicted node distances toward the real values.

There are two advantages of using $IL^{DST}_{i,j}$. First, $IL^{DST}_{i,j}$ is a continuous function of the node locations. Therefore, it can reflect the small location differences in solutions and lead to smoother optimization processes. Second, $IL^{DST}_{i,j}$ is more computationally time efficient compared with $IL^{base}_{i,j}$ because it does not require the computation of predicted hop-counts.

\subsection{Total Hop Loss with High Accuracy and Optimized Efficiency}
Combining $AC^{CC}_{i,j}$ and $IL^{DST}_{i,j}$, the proposed hop loss, $HL^{DCC}$, is defined as
\begin{equation}\label{eq:hl_ccdst}
HL^{DCC} = \sum_{i=1}^{N}\sum_{j=1}^{N}IL^{DST}_{i,j}\cdot AC^{CC}_{i,j},
\end{equation}
where DCC stands for \emph{distance-based connectivity consistency}.

As a combination of $AC^{CC}_{i,j}$ and $IL^{DST}_{i,j}$, the advantages of $HL^{DCC}$ include:
\begin{itemize}
    \item It guarantees a coverage of all hop errors in either a direct or indirect way, making all hop errors effectively penalized.
    \item The loss function is continuous regarding the node locations, making the location differences more recognizable, and the optimization processes smoother.
    \item The computation of the hop loss does not depend on calculating the predicted hop-counts, making the algorithm time efficient.
\end{itemize}


%% file: images/ac_base_fail_case/ac_base_fail_case.tex


\begin{figure}[h!]
    \centering
    \subfigure[real network]{\includegraphics[width=0.276\textwidth]{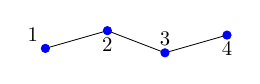}}
    \quad 
    \subfigure[prediction]{\includegraphics[width=0.156\textwidth]{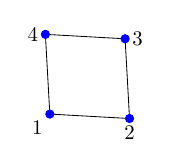}}
    \hspace{10mm}

    \caption{An example network with a hop-count inconsistency between $Hop^{real}_{1,4}$ and $Hop^{pred}_{1,4}$. Each dot represents a node, and each segment between nodes represents the nodes are connected. 
    $HL^{base}$ cannot catch this hop error because $AC^{base}_{1,4}=0$. On the other hand, $AC^{CC}_{1,4}=1$, making the modified hop loss capable of catching the error.}
    \label{fig:ac_base_fail_case}
\end{figure}

%% file: images/ILDST/ILDST.tex
\begin{figure}[h!]
    \centering
    \hfill
    \subfigure[$Dist_{i,j}^{pred} > R \geq Dist_{i,j}^{real}$]{\includegraphics[width=0.35\textwidth]{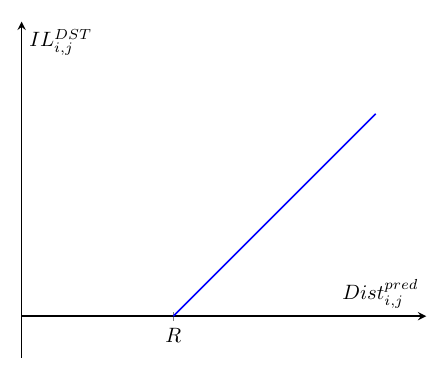}}
\hfill  
    \subfigure[$Dist_{i,j}^{pred} \leq R < Dist_{i,j}^{real}$]{\includegraphics[width=0.35\textwidth]{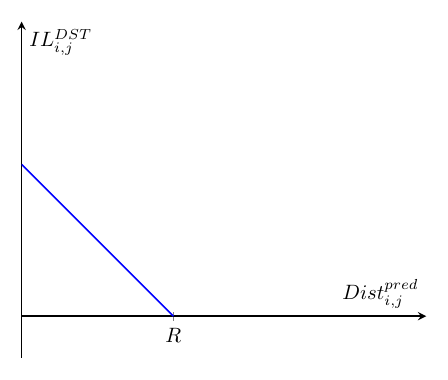}}
    \hspace{10mm}

    \caption{The relation between $IL_{i,j}^{DST}$ and $Dist_{i,j}^{pred}$ in the two conditions.}
    \label{fig:ildst}
\end{figure}

%% file: sections/experiments.tex
\section{Simulation Results}\label{sec:exp}

In this section, we provide the simulation results of the proposed algorithm and compare them with the results of other highly regarded algorithms.

\subsection{Data, Algorithm Settings and Evaluation Metric}
\begin{figure}[h]
	\center
	\subfigure[Random]{\includegraphics[scale=0.25]{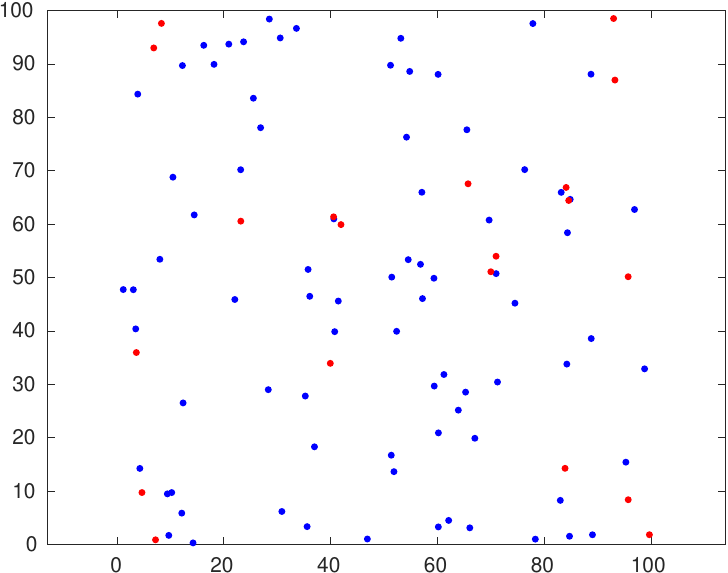}}
	\subfigure[C-shaped]{\includegraphics[scale=0.25]{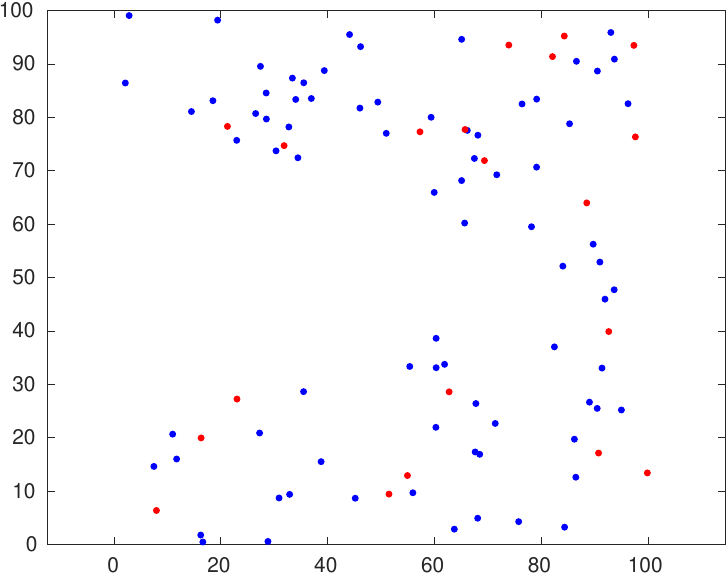}}
	\subfigure[O-shaped]{\includegraphics[scale=0.25]{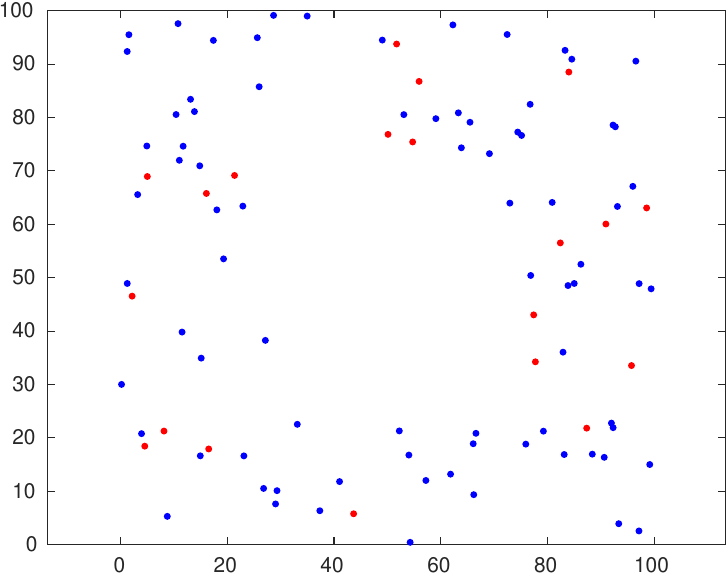}}
	\subfigure[X-shaped]{\includegraphics[scale=0.25]{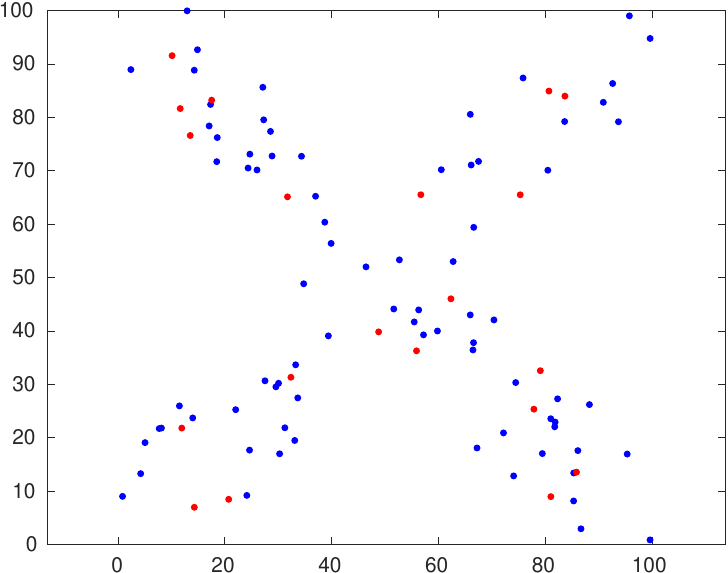}}
	
	\caption{The types of the network topology in simulations: randomly distributed, C-shaped, O-shaped and X-shaped. Each blue dot \textcolor{blue}{$\bullet$} represents an unknown node, and each red dot \textcolor{red}{$\bullet$} represents an anchor nodes.}
	\label{fig:four_scenes}
\end{figure}

\input{tables/algo_settings.tex}

\paragraph{Simulation Data} We follow the dataset used in \cite{wang2022distance, wang2024dv}, etc. In the dataset, a network is in a $100 \times 100$ square region. To evaluate the performance of the algorithm in both regular areas and irregular areas, the simulations run on four types of network topology showed in Fig. \ref{fig:four_scenes}: randomly distributed, C-shaped, O-shaped and X-shaped. In each network, there are 100 nodes. The first $N_a$ nodes, which are randomly distributed, are selected as anchor nodes, and the rest of the nodes are unknown nodes. Multiple communication radius $R$ are also experimented. See the range of $N_a$ and $R$ in TABLE \ref{tab:exp_param}.

\paragraph{Algorithm Settings} The evaluated algorithm uses multi-objective optimization. It has one objective as $HL^{DCC}$ defined in (\ref{eq:hl_ccdst}), and the other objective as the DEMN distance loss proposed in \cite{wang2024dv}. It uses the same multi-objective genetic algorithm as that in \cite{wang2024dv} for optimization. The detailed settings are in TABLE \ref{tab:exp_param}. The evaluated algorithm is referred to as DCC in the followings.

\paragraph{Evaluation Metric}
The mean localization errors ($MLEs$) are used as the evaluation metric, defined as
\begin{equation}
    MLEs = \frac{100\%}{N_{u} \cdot R} \sum_{k=1}^{N_{u}} \sqrt{(x_k^{pred} - x_k^{real})^2 + (y_k^{pred} - y_k^{real})^2},
\end{equation}
where $N_u$ is the number of unknown nodes, and $(x_k^{pred}, y_k^{pred})$, $(x_k^{real}, y_k^{real})$ are the predicted and the ground truth node location respectively for node $k$.

\subsection{Comparison Results}

We compare the $MLEs$ of the evaluated algorithm with those from other highly regarded algorithms, including \cite{niculescu2003dv09, gui2020connectivity, cui2017novel023, 8913604@023_1, OUYANG@023_2, cai2019multi024, wang2022distance, wang2024dv}. Among the compared algorithms, DEMN-DV-Hop \cite{wang2024dv} has the same settings as ours except for the design of the hop loss. Therefore, comparing with DEMN-DV-Hop directly highlights the specific improvements introduced by the hop loss modifications, and DEMN-DV-Hop is used as the baseline for more detailed evaluations.

TABLE \ref{tab_position_error_random}, \ref{tab_position_error_Cshape}, \ref{tab_position_error_Oshape}, \ref{tab_position_error_Xshape} list the results in scenarios with varying network topology, $N_a$ and $R$. In the results, DCC achieves meaningful improvements compared with other algorithms, and has the best accuracy in most of the cases. Compared with DEMN-DV-Hop, the improvement of DCC is particularly noticeable in cases where $N_a$ is small and $R$ is large. For example, in the randomly distributed network, when $N_a=5, R=40$, the relative improvement is as large as 48\%. This is because when $N_a$ is small, the information from the distance estimation is limited, and when $R$ is large, nodes can have more first-order connectivity relationships, which can be reflected in the hop loss. The evaluation results indicate that DCC exploits those connectivity information more effectively than DEMN-DV-Hop. 

\begin{figure}[h]
	\center
	\includegraphics[scale=0.18]{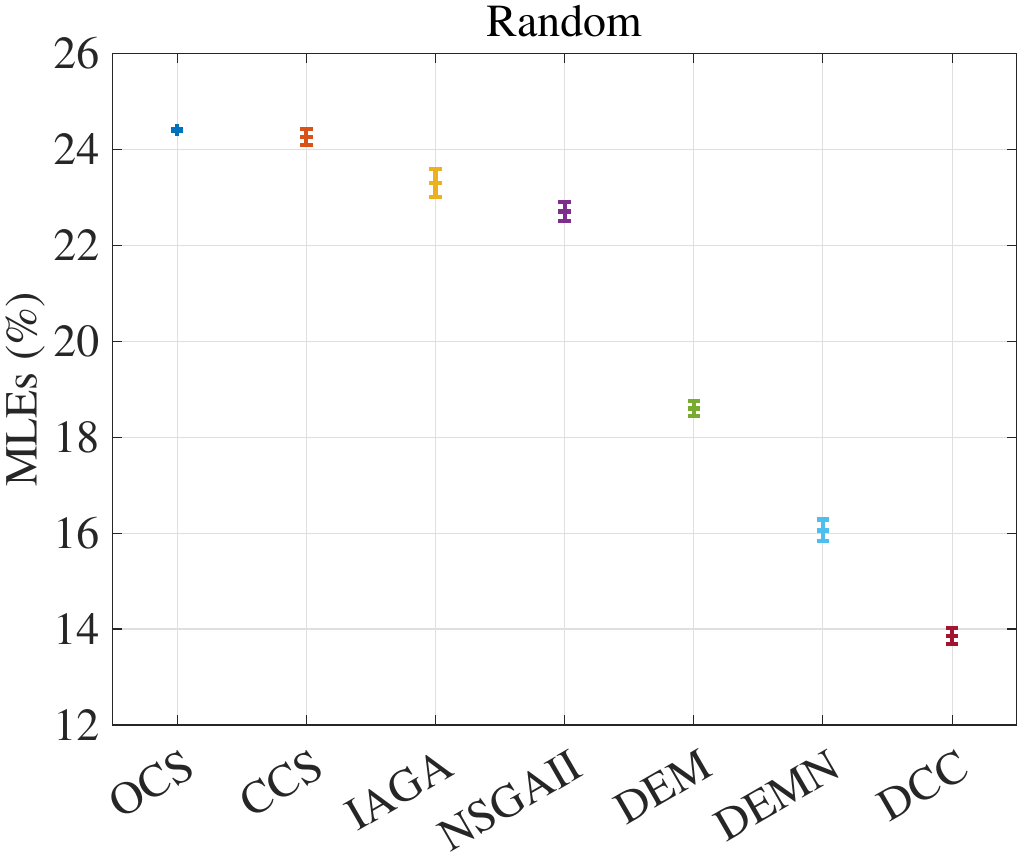}
	\includegraphics[scale=0.18]{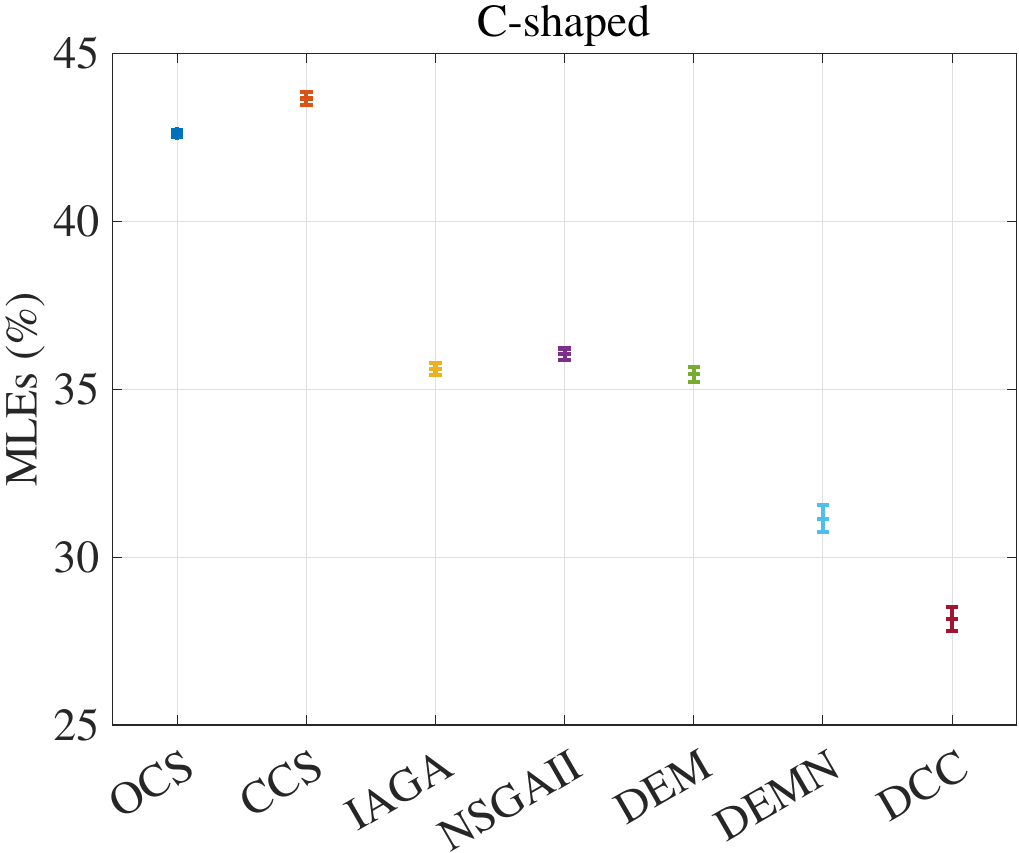}
        \includegraphics[scale=0.18]{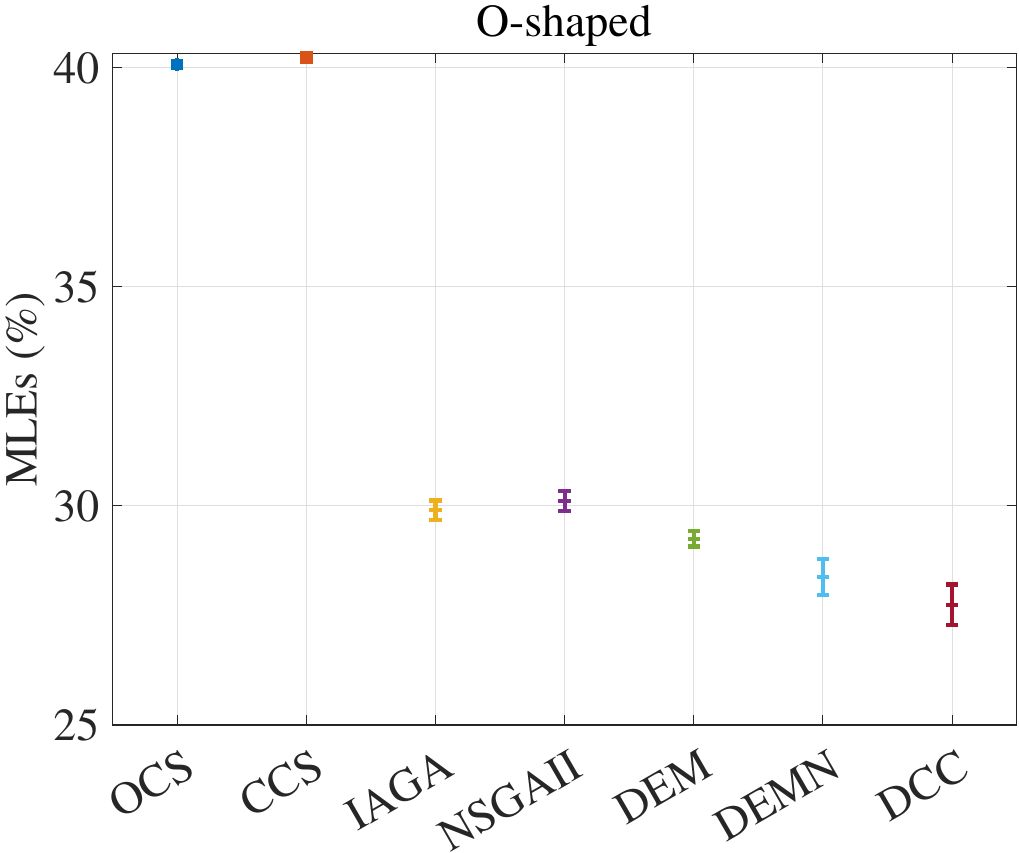}
	\includegraphics[scale=0.18]{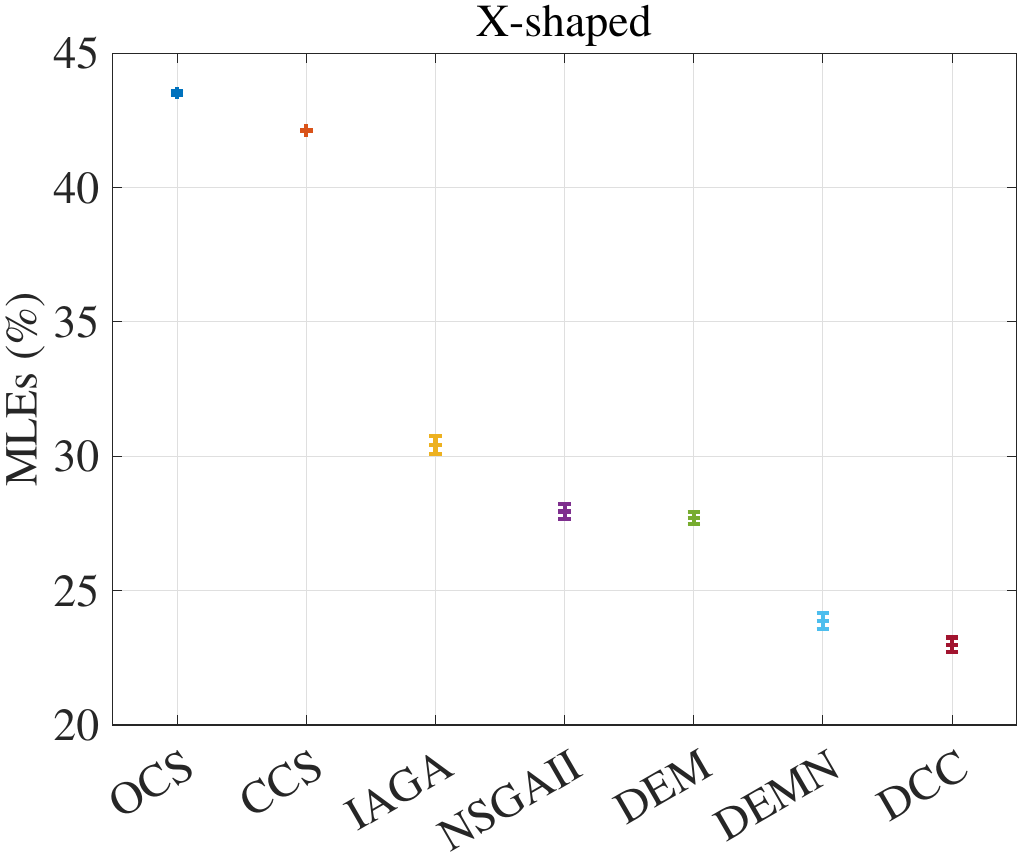}
	
	\caption{The $ 95\% $ confidence interval of $ MLEs$. $ {N}_{a}=20 $, $  R=25 $. }
	\label{fig:confidence_interval}
\end{figure}

\input{tables/comparison/random_comparison.tex}

\input{tables/comparison/Cshape_comparison.tex}

\input{tables/comparison/Oshape_comparison.tex}

\input{tables/comparison/Xshape_comparison.tex}

To analyze the distribution of the evaluation results, Fig. \ref{fig:confidence_interval} displays the 95\% confidence intervals (CI) for the scenarios where $N_a=20, R=25$. The midpoint of the CI represents the mean value of $MLEs$, and the width of the CI is related to the standard deviation of the data and the confidence level. The CIs demonstrate that the improvement of DCC is significant.

Fig. \ref{fig:sol_est_vs_gt} displays the graphs of example solutions with the location error vectors obtained by DEMN-DV-Hop and DCC. The solutions of DCC have particularly significant improvements in certain areas, which are highlighted by green boxes in the graphs.

\input{images/localization_error/LE_and_CDF.tex}

The time complexities are evaluated for DCC and DEMN-DV-Hop. 
In the simulations, DEMN-DV-Hop uses efficient breadth-first search algorithm to compute the predicted hop-counts. On the other hand, DCC computes the predicted distances between nodes instead of the hop-counts.
Fig. \ref{fig:time_profiling} plots the total time cost of DCC and DEMN-DV-Hop to solve the problems in randomly distributed networks with varying $N_a$ and $R$. According to the plot, the total time cost is reduced by approximately 30\% to 40\% by using DCC.

\begin{figure}[h]
    \centering
    \includegraphics[width=0.28\linewidth]{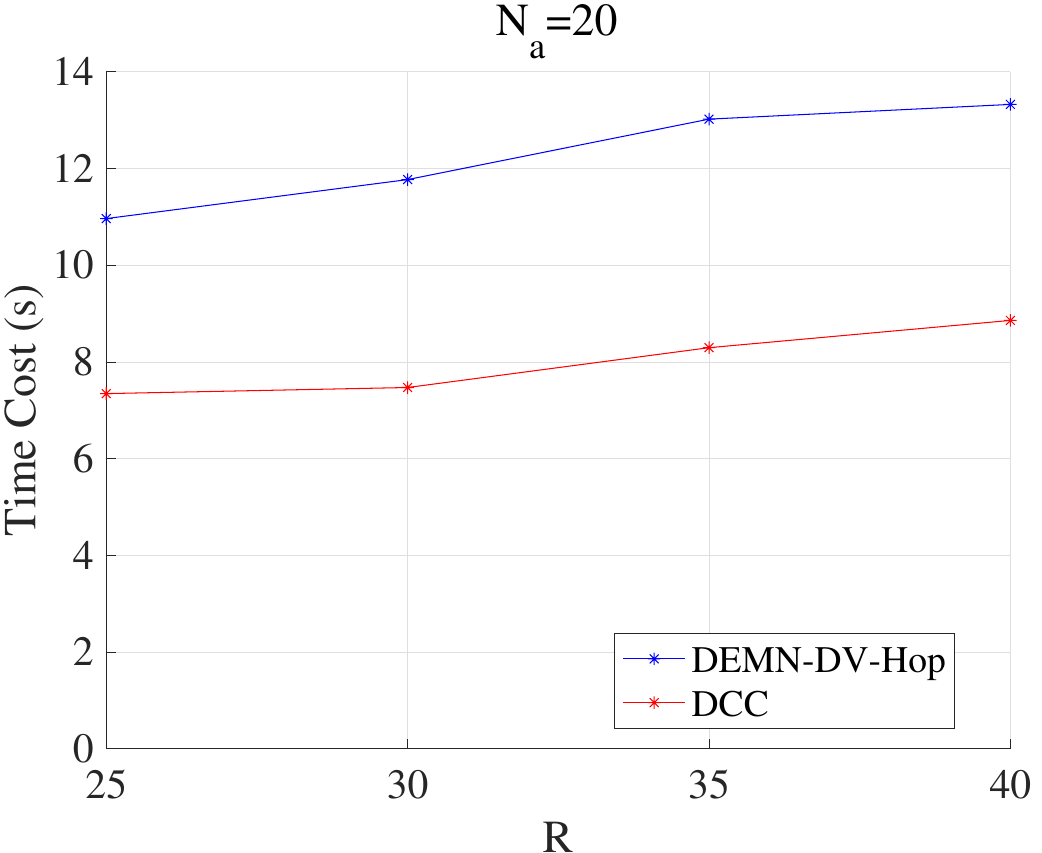}
    \includegraphics[width=0.28\linewidth]{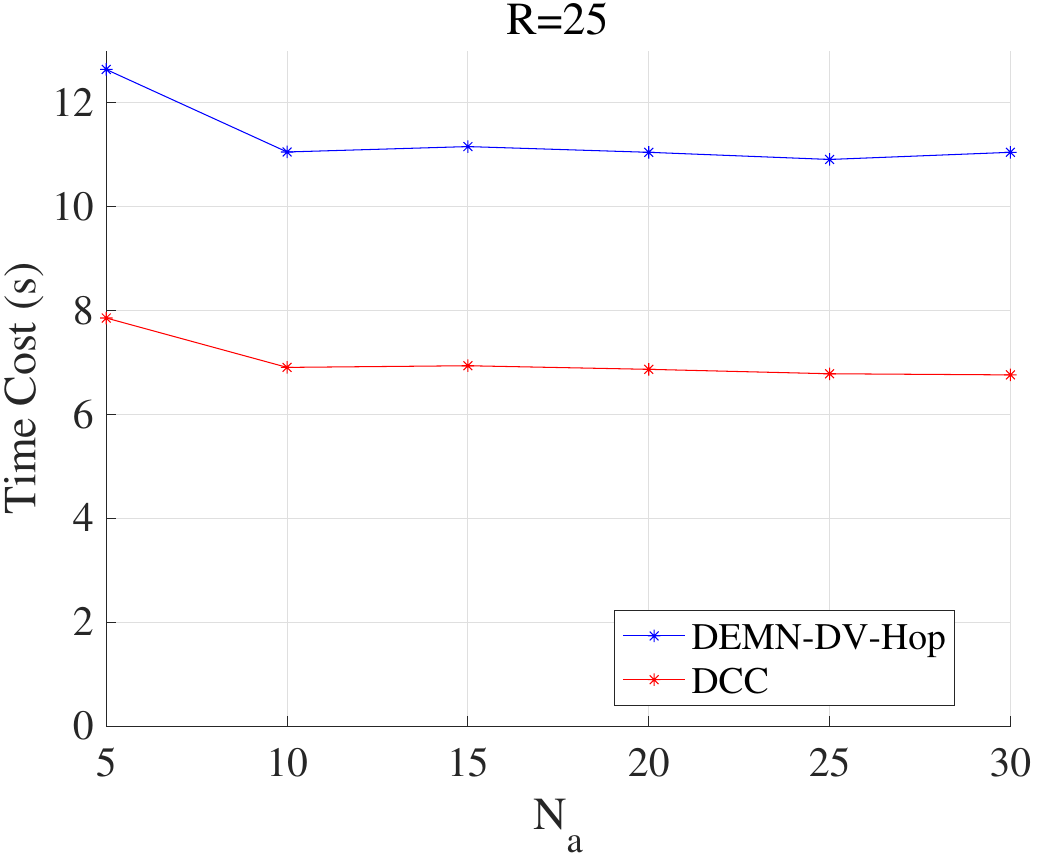}
    \caption{Time cost of DEMN-DV-Hop and DCC in randomly distributed networks. Left: $N_a=20$, $R=25, 30, 35, 40$. Right: $R = 25$, $N_a = 5, 10, 15, 20, 25, 30$.}
    \label{fig:time_profiling}
\end{figure}

\subsection{Ablation Study}
The ablation study is conducted in this subsection to assess the impact of $AC^{CC}$ and $IL^{DST}$ respectively. In the ablation study, we compare three multi-objective algorithms with same optimization settings. The three algorithms all have the first objective as the DEMN distance loss. They have different hop losses as the second objective:
\begin{enumerate}
    \item DEMN-DV-Hop: the hop loss objective is $HL^{base}$ defined in (\ref{eq:base_hl}).
    \item ACCC (activation condition based on connectivity consistency): the hop loss is
    \begin{equation}
        HL^{ACCC} = \sum_{i=1}^N \sum_{j=1}^N IL^{base}_{i,j}\cdot AC^{CC}_{i,j},
    \end{equation}
    where $IL^{base}_{i,j}$ is defined in (\ref{eq:il_base}) and $AC^{CC}_{i,j}$ is defined in (\ref{eq:ac_CC}).
    \item DCC: the hop loss is $HL^{DCC}$ defined in (\ref{eq:hl_ccdst}).
\end{enumerate}

The results of the ablation study for randomly distributed networks are in TABLE \ref{tab_position_error_random_ablation}. In most cases in the results, DCC has the best accuracy. This indicates the necessity of using $IL^{DST}_{i,j}$ to boost the performance. And because  using $AC^{CC}_{i,j}$ provides the necessary conditions to apply the continuous individual loss $IL^{DST}_{i,j}$, the ablation study demonstrates that both $AC^{CC}_{i,j}$ and $IL^{DST}_{i,j}$ are necessary in DCC.

\input{tables/ablation/ablation_random.tex}

%% file: tables/algo_settings.tex
\begin{table}[h]
	\renewcommand{\arraystretch}{1.03}
	\caption{Simulation parameters. $ Pc $: crossover probability,   $ Pm $: mutation probability, $ TN $: total nodes, $ N_{_{a}} $: the number of anchor nodes}
	\label{tab:exp_param}
	\centering
	\begin{tabular}{l@{\hspace{20pt}}|l|l@{\hspace{20pt}}|l}
		\hline
		\hline
		\textbf{ Parameters} & \textbf{Value} & \textbf{ Parameters} & \textbf{Value}    \\
		
		\hline
		
		$ TN $  & 100 & $ N_{_{a}} $     & 5-30 \\
		
		$ R  $     & 25-40  & Detection area      & 100m×100m\\

            Independent repeat test    & 50 & $ Pc $      & 0.9\\
            $ Pm $    & 0.1 & Population size        & 20 \\

            Variable dimension ($ v $)        & 2($ TN-N_{_{a}} $) & Maximum iterations       & 500\\
            
		\hline
		\hline
	\end{tabular}
\end{table}

%% file: tables/comparison/random_comparison.tex
\begin{table}[h!]
	\renewcommand{\arraystretch}{1.01}
	\caption{The comparison of the $MLEs$ (\%) in the random network.}
	\label{tab_position_error_random}
	\centering
	
	\begin{tabular}{ccccccccc}
		\hline
		\hline
		\multicolumn{1}{c}{$ {N}_{a} $} & \multicolumn{4}{c}{5} & \multicolumn{4}{c}{10}\\  
		\cmidrule(lr){1-1}
		\cmidrule(lr){2-5}
		\cmidrule(lr){6-9}
		
		Communication radius ($ R $)	&   25	& 	30	& 	35	& 	40	& 	25	& 	30	& 	35	& 	40  \\
		\hline
		
		DV-Hop\cite{niculescu2003dv09}	&    80.12	&   195.2	&  	142.6	& 	122.1	& 	38.41	& 	32.75	& 	33.08	& 	30.55	  \\
		
		CC-DV-Hop\cite{gui2020connectivity}	&    \textbf{40.61} & 122.8  & 83.35  & 85.22 & 34.31 &  30.53 &  29.14 &  27.07 \\
		
		OCS-DV-Hop\cite{cui2017novel023} &   50.02 &	71.04& 	76.08& 	76.52& 	31.65& 	26.47& 	30.96& 	26.82    \\
		
		CCS-DV-Hop\cite{8913604@023_1}	&    50.17 & {69.12} &	73.54 &	76.49 &	31.65 &	25.51 &	29.34 &	25.60  \\
		
		IAGA-DV-Hop\cite{OUYANG@023_2}	&   102.4	&  92.81	& 82.74	& 75.33 &	30.34 &	27.35  &	29.90  & 26.01   \\
		
		NSGAII-DV-Hop\cite{cai2019multi024} &	91.05 & 	86.90 & 	71.29 & 	63.36 & 	30.14 & 	26.85 & 	29.27 & 	25.60  \\
		
		DEM-DV-Hop\cite{wang2022distance} &	85.00 & 	82.20 & 	70.03 & 	59.64 & 	27.54 & 	23.72 & 	26.56 & 	24.12  \\

            DEMN-DV-Hop \cite{wang2024dv} & 80.82 & 76.98 & 59.77 & 43.95 & 19.59 & 16.60 & 17.77 & 14.53  \\
            
		\textbf{DCC} & 67.02 & \textbf{60.66} & \textbf{34.90} & \textbf{22.84} & \textbf{17.27} & \textbf{14.35} & \textbf{15.47} & \textbf{11.92}  \\
  
		\hline
		\hline
  
		\multicolumn{1}{c}{$ {N}_{a} $} & \multicolumn{4}{c}{15} & \multicolumn{4}{c}{20} \\ 
		\cmidrule(lr){1-1}
		\cmidrule(lr){2-5}
		\cmidrule(lr){6-9}
		
		Communication radius ($ R $)	&   25	& 	30	& 	35	& 	40	& 	25	& 	30	& 	35	& 	40 \\
		\hline
		
		DV-Hop\cite{niculescu2003dv09}	& 	28.65	& 	29.06	& 	32.86	& 	27.09   &    31.98	& 	29.12	& 	28.19	& 	26.16  \\
		
		CC-DV-Hop\cite{gui2020connectivity}	& 23.67 &  27.09 &  28.89 &  23.47 &    26.71  & 24.37  & 23.73 &  22.48 \\
		
		OCS-DV-Hop\cite{cui2017novel023} & 	24.36& 	21.70& 	27.04& 	23.16 &   24.40	& 	21.91	& 	20.94	& 	22.62 \\
		
		CCS-DV-Hop\cite{8913604@023_1}	&	24.35 &	21.28 &	26.40 &	22.82 &    24.25	& 	22.10	& 	20.99	& 	22.13 \\
		
		IAGA-DV-Hop\cite{OUYANG@023_2}	&	24.72 &	21.41 &	25.83 &	21.06 &   23.29 &	20.22 &	21.35 &	20.39   \\
		
		NSGAII-DV-Hop\cite{cai2019multi024} & 	23.68 & 	22.10 & 	25.32 & 	20.99  &	22.70 & 	21.39 & 	21.06 & 	20.10  \\
		
		DEM-DV-Hop\cite{wang2022distance} & 	20.77 & 	19.92 & 	22.30 & 	19.30 &	18.59 & 	19.49 & 	19.26 & 	17.78 \\ 

            DEMN-DV-Hop \cite{wang2024dv} &  16.49 & 15.16 & 15.00 & 12.31 & 16.03 & 13.84 & 13.03 & 11.13 \\
            
            \textbf{DCC} & \textbf{14.12} & \textbf{13.02} & \textbf{13.38} & \textbf{10.55} & \textbf{13.80} & \textbf{11.95} & \textbf{10.77} & \textbf{9.26} \\ 
		
		\hline
		\hline

  \multicolumn{1}{c}{$ {N}_{a} $} & \multicolumn{4}{c}{25} & \multicolumn{4}{c}{30} \\ 
		\cmidrule(lr){1-1}
		\cmidrule(lr){2-5}
		\cmidrule(lr){6-9}
		
		Communication radius ($ R $)	&   25	& 	30	& 	35	& 	40	& 	25	& 	30	& 	35	& 	40 \\
		\hline
		
		DV-Hop\cite{niculescu2003dv09}	& 	27.77	& 	27.36	& 	26.38	& 	26.10	& 	33.78	& 	29.63	& 	30.87	& 	26.42     \\
		
		CC-DV-Hop\cite{gui2020connectivity}	 & 22.41  & 23.61 &  21.45 &  22.98 &  27.65 & 26.78 &  25.33 &  20.83\\
		
		OCS-DV-Hop\cite{cui2017novel023} & 	20.42	& 	18.75	& 	19.85	& 	19.02	& 	21.89	& 	18.65	& 	22.27	& 	19.26   \\
		
		CCS-DV-Hop\cite{8913604@023_1} & 	21.23	& 	19.10 	&	20.16 	&	19.13	& 	22.29	& 	19.09 	&	21.85	& 	19.36  \\
		
		IAGA-DV-Hop\cite{OUYANG@023_2}	&	20.38 &	18.60 &	19.47 &	18.77 &	21.34 &	18.27 &	19.64 &	18.50   \\
		
		NSGAII-DV-Hop\cite{cai2019multi024}  & 	19.92 & 	18.94 & 	19.40 & 	18.34 & 	20.48 & 	18.66 & 	19.63 & 	18.30  \\
		
		DEM-DV-Hop\cite{wang2022distance} & 	18.24 & 	17.01 & 	18.32 & 	16.87 & 	18.21 & 	16.78 & 	17.59 & 	16.41  \\ 
            
            DEMN-DV-Hop \cite{wang2024dv} & 14.01 & 12.29 & 11.83 & 10.36 & 14.02 & 12.45 & 11.45 & 9.95 \\

            \textbf{DCC}  & \textbf{12.32} & \textbf{10.64} &  \textbf{9.66} & \textbf{8.62} & \textbf{12.24} & \textbf{10.81} & \textbf{9.14} & \textbf{8.37} \\ 
		
		\hline
		\hline
	\end{tabular}
\end{table}

%% file: tables/comparison/Cshape_comparison.tex
\begin{table}[h!]
	\renewcommand{\arraystretch}{1.01}
	\caption{The comparison of the $MLEs$ (\%) in the C-shaped network.}
	\label{tab_position_error_Cshape}
	\centering
	\begin{tabular}{ccccccccc}
		\hline
		\hline
		\multicolumn{1}{c}{$ {N}_{a} $} & \multicolumn{4}{c}{5} & \multicolumn{4}{c}{10}\\  
		\cmidrule(lr){1-1}
		\cmidrule(lr){2-5}
		\cmidrule(lr){6-9}
		
		Communication radius ($ R $)	&   25	& 	30	& 	35	& 	40	& 	25	& 	30	& 	35	& 	40  \\
		\hline
		
		DV-Hop\cite{niculescu2003dv09}	&    88.45	&   120.1	&  	104.5	& 	90.64	& 	66.83	& 	93.20	& 	81.80	& 	66.76	  \\
		
		CC-DV-Hop\cite{gui2020connectivity}	&    83.98 & 73.33  & 64.06  & 63.87 & 52.25 &  47.20  &  46.30 &  38.77 \\
		
		OCS-DV-Hop\cite{cui2017novel023} &   \textbf{47.77} &	{40.02}   & 	{37.57}   & 	36.71   & 	42.26   & 	35.14   & 	35.29   & 35.48    \\
		
		CCS-DV-Hop\cite{8913604@023_1}&   50.09 &	41.14 &	38.41
		&  {34.29} &	42.97 &	32.88 &	33.34 &	33.02  \\
		
		IAGA-DV-Hop\cite{OUYANG@023_2}	&   116.3	&  103.6	& 86.13	& 78.20 &	41.76 &	32.55  &	32.97  & 31.71  \\
		
		NSGAII-DV-Hop\cite{cai2019multi024} &	69.82	& 	58.85	& 	58.03	& 	54.38	& 	39.33	& 	31.76	& 	28.24	& 	28.91  \\
		
		DEM-DV-Hop\cite{wang2022distance} &	62.34  &    53.11 & 	56.83 & 	51.81 &    40.21    & 	31.55 & 	27.69 & 	27.37  \\ 

            DEMN-DV-Hop\cite{wang2024dv} & 57.69 & 43.35 & 35.78 & 33.37 & 34.55 & 27.32 & 21.63 & 21.21  \\
  
            \textbf{DCC}  & 49.74 & \textbf{36.13} & \textbf{29.39} & \textbf{28.51} & \textbf{31.98} & \textbf{25.55} & \textbf{19.20} & \textbf{19.34}\\
		
		\hline
		\hline
		
		\multicolumn{1}{c}{$ {N}_{a} $} & \multicolumn{4}{c}{15} & \multicolumn{4}{c}{20} \\ 
		\cmidrule(lr){1-1}
		\cmidrule(lr){2-5}
		\cmidrule(lr){6-9}
		
		Communication radius ($ R $)	&   25	& 	30	& 	35	& 	40	& 	25	& 	30	& 	35	& 	40  \\
		\hline
		
		DV-Hop\cite{niculescu2003dv09}	& 	67.79	& 	91.93	& 	76.57   & 	70.52   &    61.31	& 	68.12	& 	59.57	& 	52.60  \\
		
		CC-DV-Hop\cite{gui2020connectivity}	& 52.76 &  50.12 &  43.18 &  45.59 &  43.91  & 38.03  & 41.07 &  37.84 \\
		
		OCS-DV-Hop\cite{cui2017novel023} & 	45.67   & 	36.35   & 	34.05   & 	33.46   &   42.61	& 	35.09	& 	35.85	& 	32.45	 \\
		
		CCS-DV-Hop\cite{8913604@023_1}	&	44.90 &	36.29 &	33.24 &	32.96 &    43.65	& 	34.97	& 	35.51	& 	32.39  \\
		
		IAGA-DV-Hop\cite{OUYANG@023_2}	& 42.97 &	34.98 &	32.95 &	33.11 &   35.59 &	29.70 &	29.84 &	28.84  \\
		
		NSGAII-DV-Hop\cite{cai2019multi024} 	& 	40.23	& 	31.26	& 	27.23	& 	28.72 &	36.04  &	30.44 & 27.23 & 	27.17  \\
		
		DEM-DV-Hop\cite{wang2022distance}  & 	38.44 & 	30.04 & 	26.35 & 	26.99 &  35.44  &	29.32   &	{25.34}   &	26.04  \\

            DEMN-DV-Hop\cite{wang2024dv} & 33.99 & 27.64 & 22.86 & 21.95 & 31.00 & 25.54 & 21.31 & 20.60 \\
  
            \textbf{DCC}  &  \textbf{30.46} & \textbf{25.03} & \textbf{20.57} & \textbf{20.22} &\textbf{28.23} & \textbf{23.53} & \textbf{19.34} & \textbf{19.31}  \\
		
		\hline
		\hline

  \multicolumn{1}{c}{$ {N}_{a} $} & \multicolumn{4}{c}{25} & \multicolumn{4}{c}{30} \\ 
		\cmidrule(lr){1-1}
		\cmidrule(lr){2-5}
		\cmidrule(lr){6-9}
		
		Communication radius ($ R $)	&   25	& 	30	& 	35	& 	40	& 	25	& 	30	& 	35	& 	40 \\
		\hline
		
		DV-Hop\cite{niculescu2003dv09}		& 	60.91	& 	59.27	& 	53.28	& 	46.08	& 	63.28	& 	53.51	& 	49.10	& 	42.31     \\
		
		CC-DV-Hop\cite{gui2020connectivity}	 & 44.02  & 39.50 &  40.02 &  34.51 &  44.88 & 40.87 & 42.16 &  35.14\\
		
		OCS-DV-Hop\cite{cui2017novel023} 	& 	54.28	& 	39.17	& 	40.97	& 36.02	& 	53.55	& 	40.38	& 	42.64	& 	37.98   \\
		
		CCS-DV-Hop\cite{8913604@023_1}	& 	50.81	& 	37.29 	&	39.86 	&	34.16	& 	53.97	& 	38.72	&	42.37	& 	37.79  \\
		
		IAGA-DV-Hop\cite{OUYANG@023_2}	&	39.66 &	31.05 &	31.81 &	29.17 &	38.63 &	30.95 &	31.99 &	28.93   \\
		
		NSGAII-DV-Hop\cite{cai2019multi024} & 	38.23 & 	30.83 & 28.11 & 	27.21  &	37.17 & 	31.52 & 	27.87 & 	26.32   \\
		
		DEM-DV-Hop\cite{wang2022distance}   &	37.23   &	29.52  & 	{25.96 }  &   25.89   &	35.98   &	29.60  &	{25.84}   &   {25.23}  \\

            DEMN-DV-Hop\cite{wang2024dv} & 30.96 & 24.98 & 21.09 & 20.04 & 29.38 & 25.02 & 20.87 & 19.21 \\
  
            \textbf{DCC}  & \textbf{26.62} & \textbf{23.03} & \textbf{18.80} & \textbf{18.50} & \textbf{25.83} & \textbf{22.59} & \textbf{18.56} & \textbf{17.21} \\
		
		\hline
		\hline
	\end{tabular}
\end{table}

%% file: tables/comparison/Oshape_comparison.tex
\begin{table}[h!]
	\renewcommand{\arraystretch}{1.01}
	\caption{The comparison of the $MLEs$ (\%) in the O-shaped network.}
	\label{tab_position_error_Oshape}
	\centering
	\begin{tabular}{ccccccccc}
		\hline
		\hline
		\multicolumn{1}{c}{$ {N}_{a} $} & \multicolumn{4}{c}{5} & \multicolumn{4}{c}{10} \\  
		\cmidrule(lr){1-1}
		\cmidrule(lr){2-5}
		\cmidrule(lr){6-9}
		
		Communication radius ($ R $)	&   25	& 	30	& 	35	& 	40	& 	25	& 	30	& 	35	& 	40  \\
		\hline
		
		DV-Hop\cite{niculescu2003dv09}	&    84.04	&   152.2	&  	121.7	& 	100.2	& 	50.68	& 	136.3	& 	109.2	& 	92.96	  \\
		
		CC-DV-Hop\cite{gui2020connectivity}	&    \textbf{41.56} & \textbf{35.11}  & 46.24  & {36.64} & 42.70 &  27.78 &  24.31 &  34.78 \\
		
		OCS-DV-Hop\cite{cui2017novel023} &   43.29 &    46.62   & 	43.39   & 	42.82   & 	37.64   & 	34.66   & 	31.86   & 	28.95    \\
		
		CCS-DV-Hop\cite{8913604@023_1}	&    41.80 &	45.16 &	43.57 &	42.69 &	37.74 &	33.09 &	30.36 &	27.84  \\
		
		IAGA-DV-Hop\cite{OUYANG@023_2}	&   62.26	&  61.53	& 56.40	& 51.91 &	34.49 &	30.54  &	26.55  & 29.53   \\
		
		NSGAII-DV-Hop\cite{cai2019multi024} &	55.11	& 	59.35	& 	46.04	& 	47.87	& 	35.77	& 	27.84	& 	24.62	& 	27.98	  \\
		
		DEM-DV-Hop\cite{wang2022distance} &	50.72  &	52.84 & 	42.33 & 	44.42 &    33.02 & 	{24.85} & 	21.96 & 	{22.57}  \\ 

            DEMN-DV-Hop \cite{wang2024dv} & 75.45 & 56.14 & 40.80 & 37.33 & \textbf{32.43} & 21.29 & 18.34 & 17.14  \\

            \textbf{DCC}  & 59.42 & 38.90 & \textbf{26.69} & \textbf{22.18} & 32.77 & \textbf{19.98} & \textbf{16.66} & \textbf{15.66}  \\
		
		\hline
		\hline
		
		\multicolumn{1}{c}{$ {N}_{a} $} & \multicolumn{4}{c}{15} & \multicolumn{4}{c}{20} \\ 
		\cmidrule(lr){1-1}
		\cmidrule(lr){2-5}
		\cmidrule(lr){6-9}
		
		Communication radius ($ R $)	&   25	& 	30	& 	35	& 	40	& 	25	& 	30	& 	35	& 	40	\\
		\hline
		
		DV-Hop\cite{niculescu2003dv09}	& 	53.49	& 	101.9	& 	80.13	& 	68.32   &    47.14	& 	92.79	& 	73.85	& 	64.74    \\
		
		CC-DV-Hop\cite{gui2020connectivity}	 & 38.49 &  28.15 &  21.56 &  26.27 &  37.39  & 27.11  & 24.26 &  26.05 \\
		
		OCS-DV-Hop\cite{cui2017novel023}  & 	38.10   & 	30.04   & 	26.93   & 	24.01   &   40.07	& 	34.53	& 	28.52	& 	24.77  \\
		
		CCS-DV-Hop\cite{8913604@023_1}	&	38.07 &	30.02 &	27.32 &	23.97 &    40.23	& 	33.09	& 	27.53	& 	23.95 \\
		
		IAGA-DV-Hop\cite{OUYANG@023_2}	&	{30.79} &	24.50 &	21.08 &	21.77 &   29.90 &	23.26 &	20.35 &	19.15  \\
		
		NSGAII-DV-Hop\cite{cai2019multi024} & 	32.53	& 	22.33	& 	20.59	& 	20.24 &	30.11  &	21.33 & 	18.92 & 	17.37  \\
		
		DEM-DV-Hop\cite{wang2022distance} & 	32.57 & 	21.42 & 	18.70 & 	{17.49}  &	29.25  &	20.87  &	{17.24}  &	{14.92}  \\ 

            DEMN-DV-Hop \cite{wang2024dv} & 30.67 & 19.81 & 16.93 & 14.28 & 28.16 & 18.71 & 15.60 & 12.93 \\

            \textbf{DCC} & \textbf{30.38} & \textbf{18.70} & \textbf{14.91} & \textbf{13.04} & \textbf{27.84} & \textbf{17.98} & \textbf{14.36} & \textbf{11.86} \\
		
		\hline
		\hline

  \multicolumn{1}{c}{$ {N}_{a} $} & \multicolumn{4}{c}{25} & \multicolumn{4}{c}{30} \\ 
		\cmidrule(lr){1-1}
		\cmidrule(lr){2-5}
		\cmidrule(lr){6-9}
		
		Communication radius ($ R $)	&   25	& 	30	& 	35	& 	40	& 	25	& 	30	& 	35	& 	40	\\
		\hline
		
		DV-Hop\cite{niculescu2003dv09}	& 	48.10	& 	83.41	& 	67.19	& 	57.83	& 	48.06	& 	47.54	& 	44.59	& 	37.29     \\
		
		CC-DV-Hop\cite{gui2020connectivity}	& 38.29  & 26.88 &  21.37 &  23.35 &  38.77 & 24.99 &  27.43 &  24.28\\
		
		OCS-DV-Hop\cite{cui2017novel023} & 	46.46	& 	35.16	& 	30.84	& 	24.87	& 	48.43	& 	38.93	& 	32.48	& 	28.96   \\
		
		CCS-DV-Hop\cite{8913604@023_1}	& 	46.32	& 	34.70 	&	30.14 	&	24.97	& 	50.08	& 	36.90 	&	32.00	& 	26.80  \\
		
		IAGA-DV-Hop\cite{OUYANG@023_2} &	32.32 &	22.45 &	20.43 &	18.94 &	31.53 &	22.25 &	21.01 &	19.00   \\
		
		NSGAII-DV-Hop\cite{cai2019multi024} & 	31.92 & 	21.22 & 19.10 & 	17.83  &	30.26 & 	20.57 & 	19.76 & 	17.94   \\
		
		DEM-DV-Hop\cite{wang2022distance}  &	30.73  &	21.03  & 	{16.83}  &	{14.51}  &	28.90  &	{20.09}  &	{17.06}  &	{15.16}  \\ 

            DEMN-DV-Hop \cite{wang2024dv} & 27.63 & 17.91 & 14.98 & 12.34 & 25.24 & 17.18 & 14.95 & 12.41 \\

            \textbf{DCC} & \textbf{25.85} & \textbf{17.14} & \textbf{13.49} & \textbf{11.18} & \textbf{23.76} & \textbf{16.46} & \textbf{13.13} & \textbf{11.18} \\
		
		\hline
		\hline
	\end{tabular}
\end{table}

%% file: tables/comparison/Xshape_comparison.tex
\begin{table}[h!]
	\renewcommand{\arraystretch}{1.01}
	\caption{The comparison of the $MLEs$ (\%) in the X-shaped network.}
	\label{tab_position_error_Xshape}
	\centering
	\begin{tabular}{ccccccccc}
		\hline
		\hline
		\multicolumn{1}{c}{$ {N}_{a} $} & \multicolumn{4}{c}{5} & \multicolumn{4}{c}{10} \\  
		\cmidrule(lr){1-1}
		\cmidrule(lr){2-5}
		\cmidrule(lr){6-9}
		
		Communication radius ($ R $)	&   25	& 	30	& 	35	& 	40	& 	25	& 	30	& 	35	& 	40	\\
		\hline
		
		DV-Hop\cite{niculescu2003dv09}	&    58.46	&   108.88	&  	95.24	& 	91.84	& 	60.38	& 	96.73	& 	86.75	& 	83.64	  \\
		
		CC-DV-Hop\cite{gui2020connectivity}	&    52.30 & 49.44  & 49.42  & 45.36 & 55.27 &  46.77 &  45.76 &  40.71\\
		
		OCS-DV-Hop\cite{cui2017novel023} &   42.85 &	49.13   & 	47.49   & 	38.15   & 	39.75   & 	40.24   & 	35.80   & 	32.52    \\
		
		CCS-DV-Hop\cite{8913604@023_1}	&    42.69  &	48.68   &	45.04   &	34.45   &	40.18   &	40.08   &	34.36   &	30.64   \\
		
		IAGA-DV-Hop\cite{OUYANG@023_2}	&   47.46	&  49.47	& 43.92	& 40.17 &	33.99 &	35.08  &	31.10   & 30.07  \\
		
		NSGAII-DV-Hop\cite{cai2019multi024} &	42.37	& 	44.45	& 	38.87	& 	34.70	& 	33.27	& 	31.89	& 	27.84	& 	27.18	 \\
		
		DEM-DV-Hop\cite{wang2022distance} &	40.43  &	42.36 & 	37.50 & 	33.24 &    32.95 & 30.90 & 	26.03 & 	25.14 \\ 

            DEMN-DV-Hop \cite{wang2024dv} & 32.46 & 33.69 & 30.59 & 25.42 & 28.13 & 23.95 & \textbf{19.30} & 18.38 \\

            \textbf{DCC} & \textbf{29.97} & \textbf{29.99} & \textbf{26.13}  & \textbf{22.41} & \textbf{27.59} & \textbf{23.76} & 19.36 & \textbf{17.43} \\
		
		\hline
		\hline
		
		\multicolumn{1}{c}{$ {N}_{a} $} & \multicolumn{4}{c}{15} & \multicolumn{4}{c}{20} \\ 
		\cmidrule(lr){1-1}
		\cmidrule(lr){2-5}
		\cmidrule(lr){6-9}
		
		Communication radius ($ R $)	&   25	& 	30	& 	35	& 	40	& 	25	& 	30	& 	35	& 	40	\\
		\hline
		
		DV-Hop\cite{niculescu2003dv09}	& 	51.12	& 	81.10	& 	70.13	& 	72.57   &    49.07	& 	78.38	& 	68.07	& 	67.90  \\
		
		CC-DV-Hop\cite{gui2020connectivity}	 & 43.97 &  38.37 &  35.97 &  36.14 &  42.36  & 40.13  & 36.76 &  38.06 \\   
		
		OCS-DV-Hop\cite{cui2017novel023}  & 	42.18   & 	41.91   & 	37.48   & 	35.86   &   43.51	& 	45.41	& 	37.52	& 	36.84 \\
		
		CCS-DV-Hop\cite{8913604@023_1}	 &	38.37   &	38.90   &	33.59   &	32.61 &    42.11	& 	43.44	& 	35.62	& 	36.04 \\
		
		IAGA-DV-Hop\cite{OUYANG@023_2}	&	30.83   &	29.51   &	26.14 &	27.25  &   30.42 &	29.60 &	27.35 &	27.01  \\
		
		NSGAII-DV-Hop\cite{cai2019multi024} & 	29.96	& 	28.18	& 	23.68	& 	24.20  &	27.94  &	26.55 & 	23.04 & 	23.18  \\
		
		DEM-DV-Hop\cite{wang2022distance} & 	29.39 & 	27.26 & 	21.81 & 	{21.27}   &	27.70  &	{25.03}  &	21.65  &	20.96  \\ 

            DEMN-DV-Hop \cite{wang2024dv} & 26.36 & 23.07 & 17.76 & 17.24  & 23.85 & \textbf{20.91} & 17.38 & 17.06  \\

            \textbf{DCC}  & \textbf{25.60} & \textbf{22.22} & \textbf{17.14} & \textbf{15.67}  & \textbf{22.92} & 20.95 & \textbf{16.20} & \textbf{15.34}\\
		
		\hline
		\hline

  \multicolumn{1}{c}{$ {N}_{a} $} & \multicolumn{4}{c}{25} & \multicolumn{4}{c}{30} \\ 
		\cmidrule(lr){1-1}
		\cmidrule(lr){2-5}
		\cmidrule(lr){6-9}
		
		Communication radius ($ R $)	&   25	& 	30	& 	35	& 	40	& 	25	& 	30	& 	35	& 	40 \\
		\hline
		
		DV-Hop\cite{niculescu2003dv09}	& 	52.52	& 	69.58	& 	61.29	& 	56.98	& 	54.51	& 	56.82	& 	49.58	& 	47.62     \\
		
		CC-DV-Hop\cite{gui2020connectivity}	 & 45.50  & 44.27 &  40.72 &  39.52 &  45.83 & 44.09 &  39.54 &  39.61\\   
		
		OCS-DV-Hop\cite{cui2017novel023} 	& 	45.50	& 	45.87	& 	39.09	& 	42.20	& 	45.51	& 	49.90	& 	42.38	& 	45.61   \\
		
		CCS-DV-Hop\cite{8913604@023_1}	& 	46.39	& 	46.22 	&	38.32 	&	40.86	& 	44.27	& 	48.19 	&	41.28	& 	44.55  \\
		
		IAGA-DV-Hop\cite{OUYANG@023_2}	&	29.89 &	28.34 &	26.35 &	26.90 &	26.34 &	26.89 &	25.70 &	25.96   \\
		
		NSGAII-DV-Hop\cite{cai2019multi024} & 	28.38 & 	26.49 & 22.24 & 	23.18  &	25.18 & 	24.32 & 	20.91 & 	20.91  \\
		
		DEM-DV-Hop\cite{wang2022distance} &	28.10  &	25.47  & 	20.23  &	{20.52}  &	25.05  &	23.51  &	{18.81}  &	{18.02}  \\ 

            DEMN-DV-Hop \cite{wang2024dv} & 24.65 & 21.15 & 16.51 & 16.91 & 22.82 & 19.68 & 14.47 & 15.07 \\

            \textbf{DCC}   & \textbf{23.24} & \textbf{20.49} & \textbf{15.37} & \textbf{14.58} & \textbf{21.50} & \textbf{19.04} & \textbf{13.62} & \textbf{12.99} \\
		
		\hline
		\hline
	\end{tabular}
\end{table}

%% file: images/localization_error/LE_and_CDF.tex
\begin{figure*}[h]
    \centering
\begin{tikzpicture}
\begin{scope}[scale=1]
    \newcommand*\colLabels{{"Random", "C-shaped", "O-shaped", "X-shaped"}}
    \def\rowLabels{{"DEMN \quad \quad", "\quad DCC \quad \quad"}}

    \def\totalScale{0.7}
    
    \def\nRows{2}
    \def\nCols{4}
    
    \def\imageWidth{4.5cm*\totalScale}
    \def\imageHeight{3.5cm*\totalScale}
    
    \def\startX{0*\totalScale}
    \def\startY{0*\totalScale}

    \pgfmathsetmacro{\myScale}{0.25*\totalScale}

    \node[anchor=north west] (img) at (\startX, \startY + \nRows *\imageHeight - \imageHeight/2) {\includegraphics[scale=\myScale]{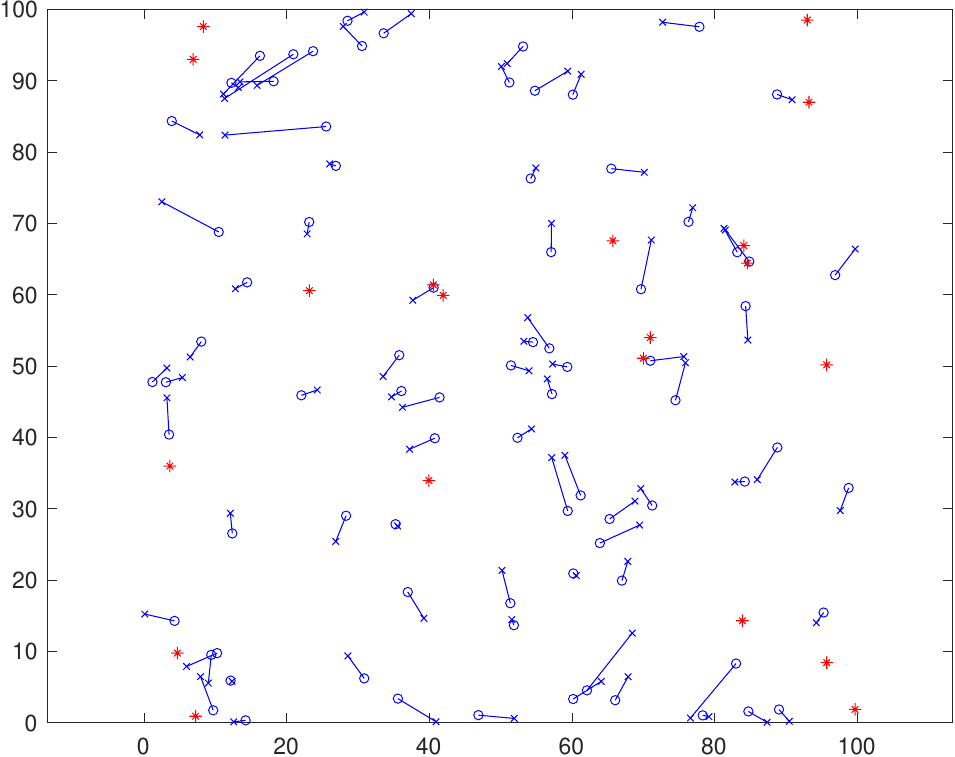}};
    \node[anchor=north west]  at (\startX + \imageWidth, \startY + \nRows *\imageHeight - \imageHeight/2) {\includegraphics[scale=\myScale]{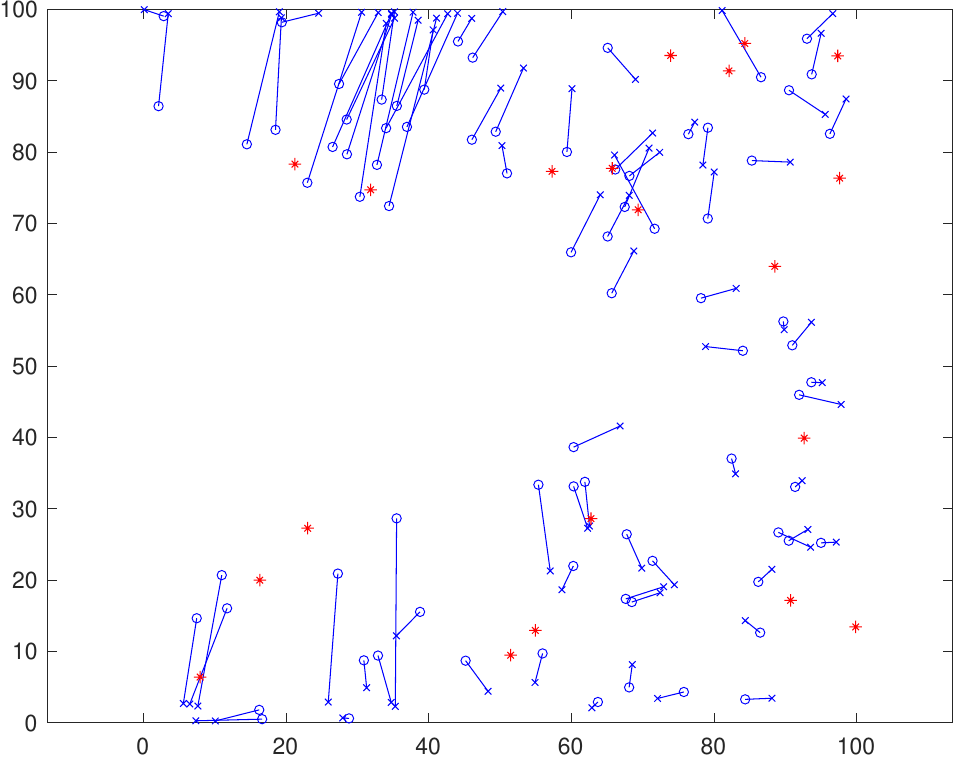}
    \begin{picture}(0,0)
        \put(-58,4){\color{green}\framebox(12,22)} 
    \end{picture}
    };
    \node[anchor=north west]  at (\startX + 2 * \imageWidth, \startY + \nRows *\imageHeight - \imageHeight/2) {\includegraphics[scale=\myScale]{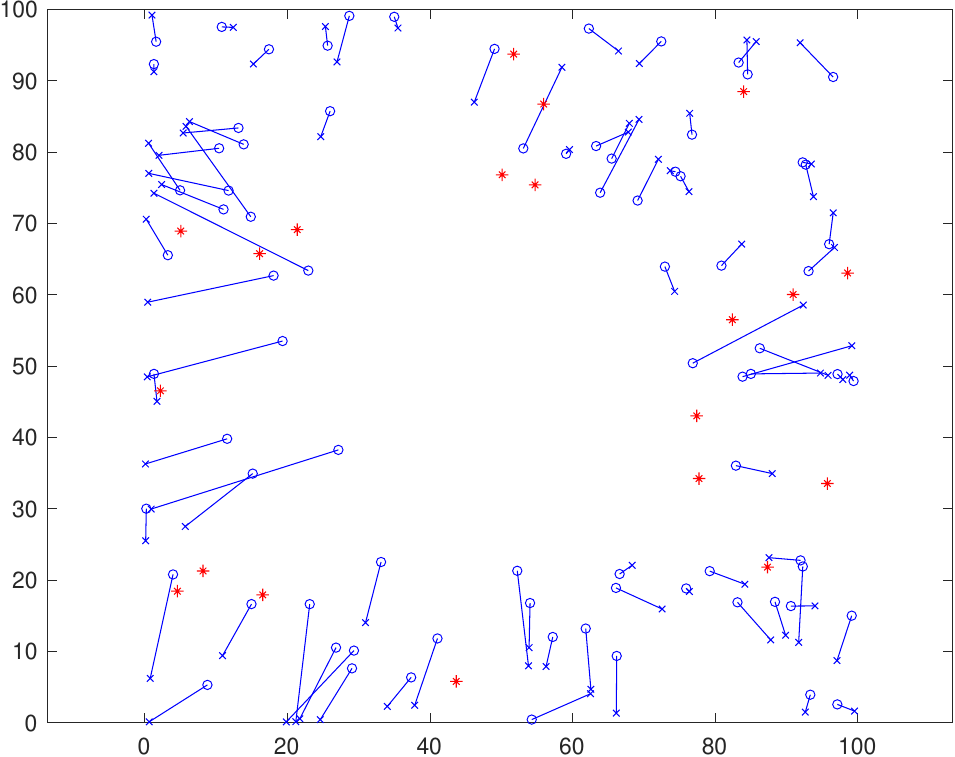}
    \begin{picture}(0,0)
        \put(-75, 18){\color{green}\framebox(22,12)} 
    \end{picture}};
    \node[anchor=north west]  at (\startX + 3 * \imageWidth, \startY + \nRows *\imageHeight - \imageHeight/2) {\includegraphics[scale=\myScale]{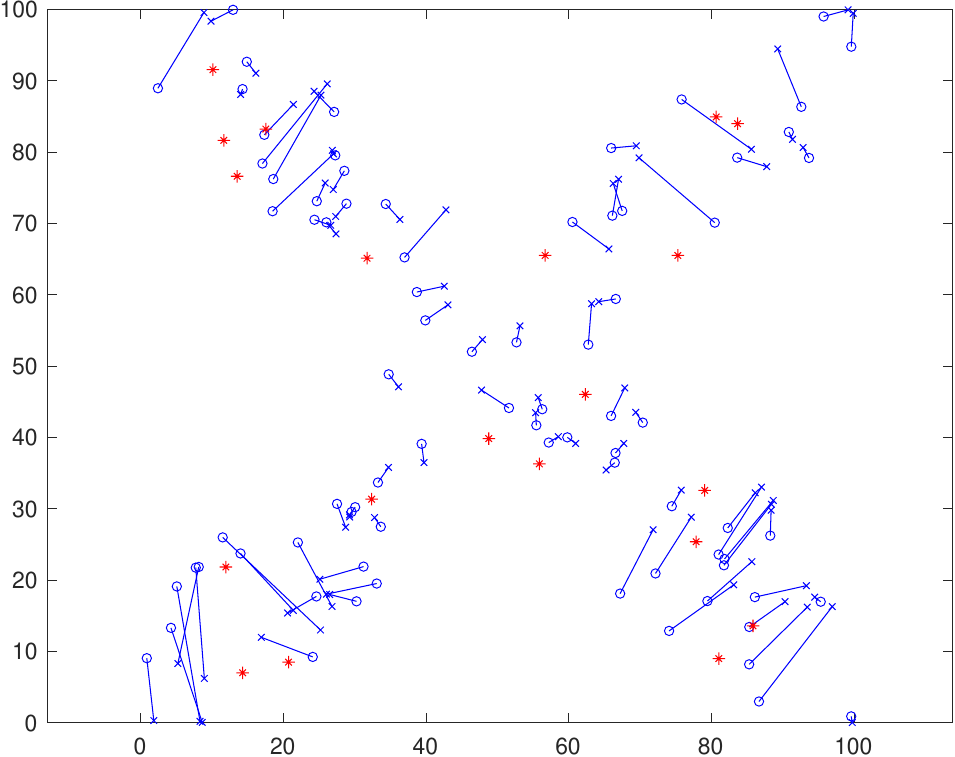}};

    \node[anchor=north west]  at (\startX, \startY + \nRows *\imageHeight - 3 * \imageHeight/2) {\includegraphics[scale=\myScale]{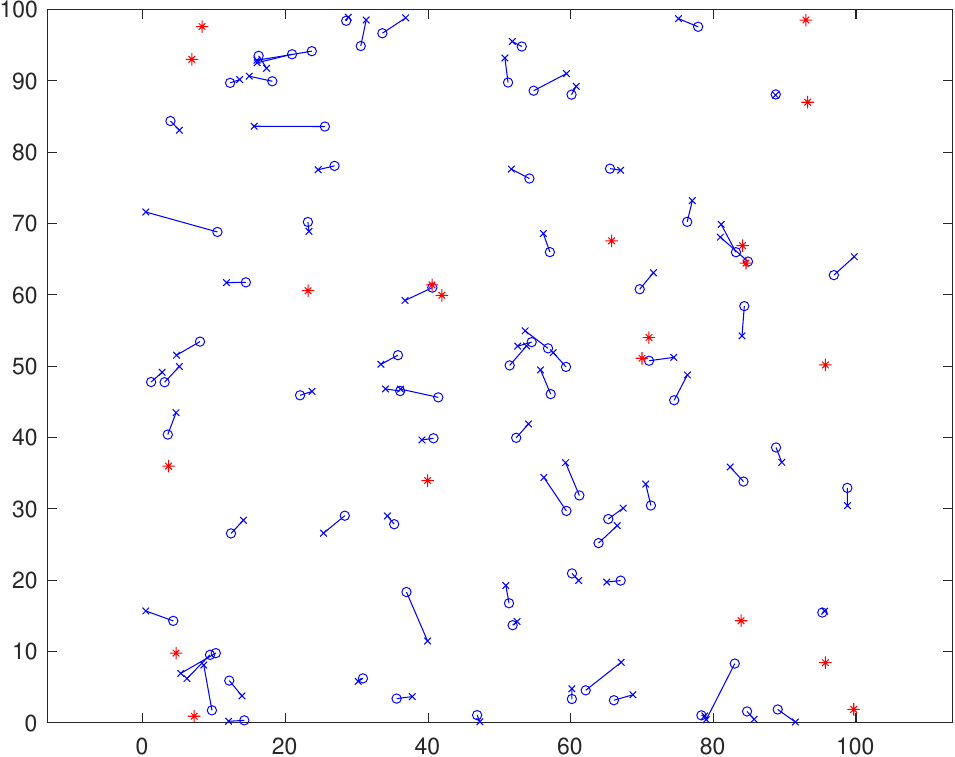}};
    \node[anchor=north west]  at (\startX + \imageWidth, \startY + \nRows *\imageHeight - 3 * \imageHeight/2) {\includegraphics[scale=\myScale]{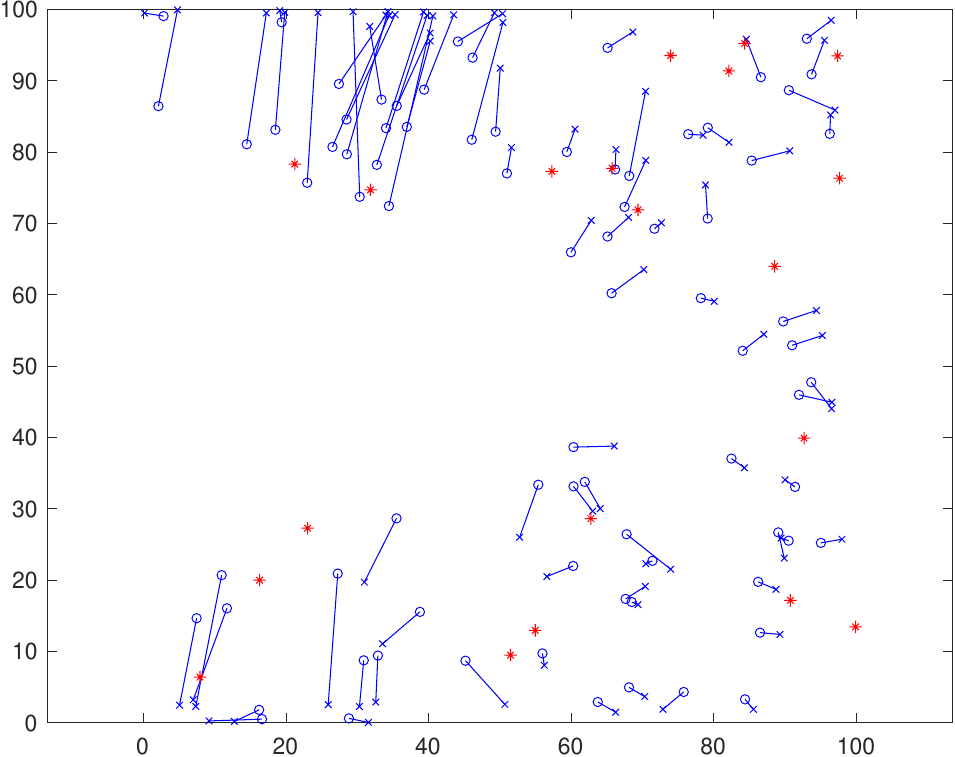}
    \begin{picture}(0,0)
        \put(-58,4){\color{green}\framebox(12,22)} 
    \end{picture}
    };
    \node[anchor=north west]  at (\startX + 2 * \imageWidth, \startY + \nRows *\imageHeight - 3 * \imageHeight/2) {\includegraphics[scale=\myScale]{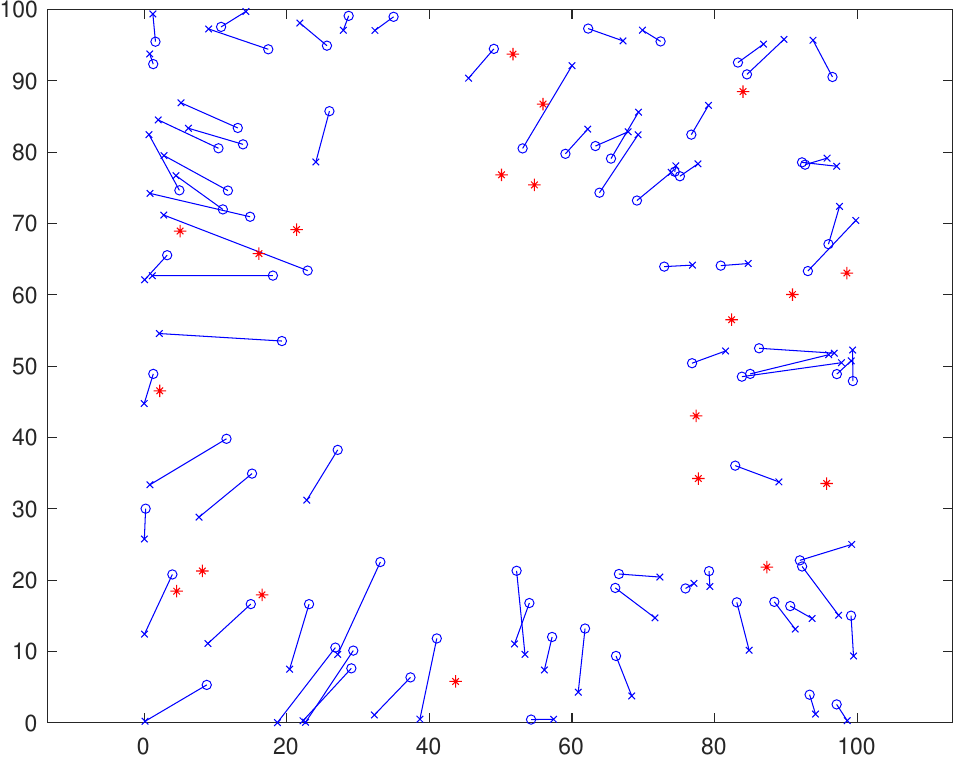}
    \begin{picture}(0,0)
        \put(-75, 18){\color{green}\framebox(22,12)} 
    \end{picture}};
    \node[anchor=north west]  at (\startX + 3 * \imageWidth, \startY + \nRows *\imageHeight - 3 * \imageHeight/2) {\includegraphics[scale=\myScale]{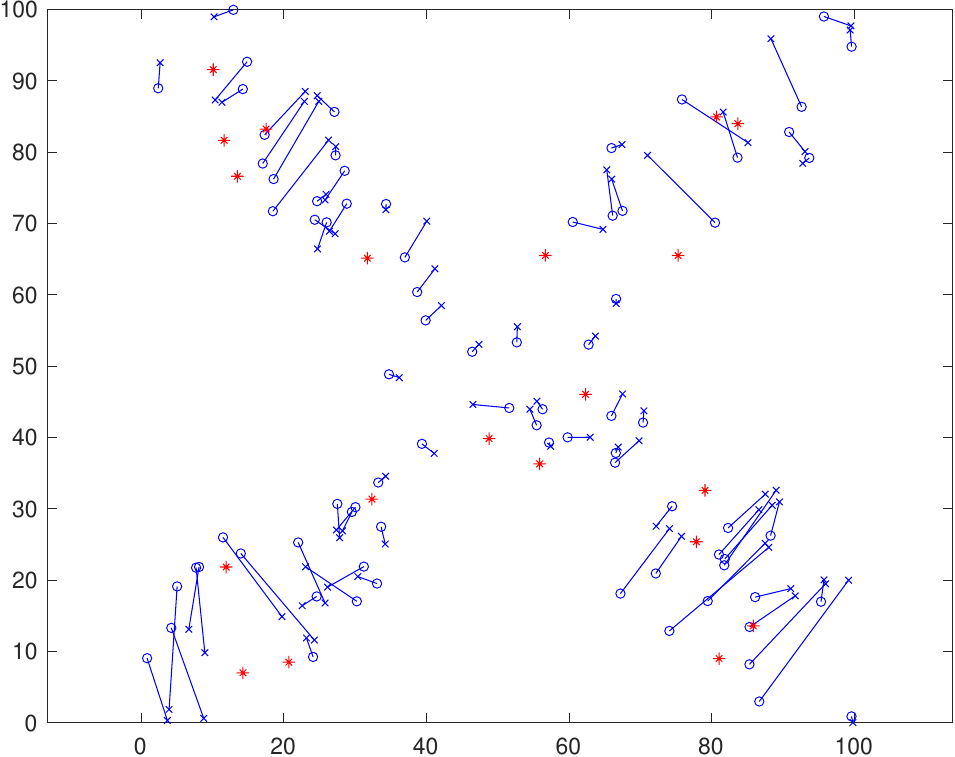}
    };

    \coordinate (horizon) at ($(img.north east) - (img.north west)$);
    \coordinate (vertical) at ($(img.north east) - (img.south east)$);
    \path (img.south west) -- (img.north east) coordinate[pos=0.5] (center)
      let \p1 = ($(img.north east) - (img.south west)$) in
      coordinate (imgwidth) at (\x1, 0)
      coordinate (imgheight) at (0, \y1);

    \foreach \col in {1,...,\nCols} {
        \node[anchor=south] at (\startX + \col*\imageWidth - \imageWidth/2, \startY + \nRows*\imageHeight - \imageHeight / 1.7) {\strut\pgfmathprint{\colLabels[\col-1]}};
        
    }
    
    \foreach \row in {1,...,\nRows} {
        \node[anchor=east,rotate=90] at (\startX , \startY + \nRows*\imageHeight - \row*\imageHeight +  \imageHeight/2.5) {\strut\pgfmathprint{\rowLabels[\row-1]}};
    }
\end{scope}
\end{tikzpicture}
\caption{The location error vector graphs for solutions from DEMN-DV-Hop and DCC when $R=25$, $N_{a} = 20$.
\protect\tikz[baseline=-0.5ex] \protect\node[star,star points=5,star point ratio=2.25,minimum size=1.2ex,inner sep=0pt,fill=red] at (0,0) {}; represents the anchor node, \protect\tikz[baseline=-0.5ex] \protect\draw[blue, thick] (0,0) circle (0.5ex); represents the real unknown node, \protect\tikz[baseline=-0.5ex] \protect\draw[blue, thick] (0.25ex,0.25ex) -- (-0.25ex,-0.25ex) (0.25ex,-0.25ex) -- (-0.25ex,0.25ex); represents the predicted unknown node, and the segment \protect\tikz[baseline=-0.5ex] \protect\draw[blue, thick] (0,0) -- (1ex,0); connects the real location and the predicted location, representing the location error.
The green boxes highlight example nodes where DCC has significant accuracy improvements.
}
\label{fig:sol_est_vs_gt}
\end{figure*}

%% file: tables/ablation/ablation_random.tex
\begin{table}[t]
	\renewcommand{\arraystretch}{1.01}
	\caption{Ablation study of the $MLEs$ (\%) in the random network.}
	\label{tab_position_error_random_ablation}
	\centering
	\begin{tabular}{ccccccccc}
		 \hline  \hline 
		 \multicolumn{1}{c}{$ {N}_{a} $} &  \multicolumn{4}{c}{5} &  \multicolumn{4}{c}{10} \\ 
		\cmidrule(lr){1-1}\cmidrule(lr){2-5}
		\cmidrule(lr){6-9}
		
		Communication radius ($ R $)	& 25 & 30 & 35 & 40 & 25 & 30 & 35 & 40  \\  \hline 
		DEMN-DV-Hop \cite{wang2024dv} & 80.82 & 76.98 & 59.77 & 43.95 & 19.59 & 16.60 & 17.77 & 14.53 \\
            ACCC & \textbf{61.85} & 61.40 & 39.08 & 29.19 & 23.17 & 16.51 & 17.11 & 13.99  \\
            DCC & 67.02 & \textbf{60.66} & \textbf{34.90} & \textbf{22.84} & \textbf{17.27} & \textbf{14.35} & \textbf{15.47} & \textbf{11.92}  \\
		
		 \hline  \hline 
		 \multicolumn{1}{c}{$ {N}_{a} $} &  \multicolumn{4}{c}{15} &  \multicolumn{4}{c}{20} \\ 
		\cmidrule(lr){1-1}\cmidrule(lr){2-5}
		\cmidrule(lr){6-9}
		
		Communication radius ($ R $)	& 25 & 30 & 35 & 40 & 25 & 30 & 35 & 40 \\  \hline 
		DEMN-DV-Hop \cite{wang2024dv} & 16.49 & 15.16 & 15.00 & 12.31 & 16.03 & 13.84 & 13.03 & 11.13 \\

            ACCC & 16.12 & 14.75 & 14.47 & 12.07 & 15.68 & 13.50 & 12.44 & 10.72 \\
            DCC & \textbf{14.12} & \textbf{13.02} & \textbf{13.38} & \textbf{10.55} & \textbf{13.80} & \textbf{11.95} & \textbf{10.77} & \textbf{9.26}  \\
		
		 \hline 
		 \hline 

   \multicolumn{1}{c}{$ {N}_{a} $} &  \multicolumn{4}{c}{25} &  \multicolumn{4}{c}{30} \\ 
		\cmidrule(lr){1-1}\cmidrule(lr){2-5}
		\cmidrule(lr){6-9}
		
		Communication radius ($ R $)	& 25 & 30 & 35 & 40 & 25 & 30 & 35 & 40 \\  \hline 
		DEMN-DV-Hop \cite{wang2024dv} & 14.01 & 12.29 & 11.83 & 10.36 & 14.02 & 12.45 & 11.45 & 9.95 \\

            ACCC  & 14.06 & 11.91 & 11.59 & 10.11 & 14.20 & 12.04 & 11.06 & 9.89 \\
            DCC &  \textbf{12.32} & \textbf{10.64} &  \textbf{9.66} & \textbf{8.62} & \textbf{12.24} & \textbf{10.81} & \textbf{9.14} & \textbf{8.37} \\
		
		 \hline 
		 \hline 
		
	\end{tabular}
\end{table}



%% file: sections/conclusion.tex
\section{Conclusion}\label{sec:conclusion}

In this paper, we propose a novel modeling of hop loss as an objective in the DV-Hop localization algorithm to enhance the performance. We aim to improve the localization accuracy as well as the time complexity. We formulate the general hop loss by introducing two key factors: AC and IL. We propose $AC^{CC}_{i,j}$ as our AC based on the first order connectivity consistency, and propose $IL^{DST}_{i,j}$ as our IL based on the distances between nodes. Combining them, we propose the new hop loss function $HL^{DCC}$. We prove that $AC^{CC}_{i,j}$ is capable of catching all hop errors in the predictions, and discuss the advantages of using a continuous individual loss. The evaluation results demonstrate that the proposed algorithm has improved accuracy compared with other DV-Hop based algorithms. And it has a significant time complexity improvement compared with DEMN-DV-Hop, the baseline algorithm using hop loss.

In future work, we will study other modelings of AC and IL. Particularly, AC can be designed as continuous weights instead of binary values, therefore more intelligently determine the importance of node pairs.


\FloatBarrier

\needspace{3\baselineskip}